\documentclass[12pt]{article}
\usepackage{graphics}
\usepackage{amssymb}

\newcommand{\e}{\mathrm{e}}
\newcommand{\D}{\mathrm{d}}
\newcommand{\C}{\mathbb{C}}
\newcommand{\N}{\mathbb{N}}
\newcommand{\R}{\mathbb{R}}

\newcommand{\Oo}{\mathcal{O}}

\newcommand{\bIe}{I^c _\epsilon}
\newcommand{\Ie}{I_\epsilon}
\newcommand{\Qle}{Q_{\lambda }^\epsilon}
\newcommand{\Ple}{P_{\lambda }^\epsilon}
\newcommand{\Qlec}{Q_{\lambda }^{\epsilon c}}
\newcommand{\Hag}{H_{\alpha,\Gamma}}
\newcommand{\Hm}[1]{\leavevmode{\marginpar{\tiny%
$\hbox to 0mm{\hspace*{-0.5mm}$\leftarrow$\hss}%
\vcenter{\vrule depth 0.1mm height 0.1mm width \the\marginparwidth}%
\hbox to
0mm{\hss$\rightarrow$\hspace*{-0.5mm}}$\\\relax\raggedright #1}}}
\newtheorem{claim}{Claim}[section]
\newtheorem{theorem}[claim]{Theorem}
\newtheorem{lemma}[claim]{Lemma}

\newtheorem{remark}[claim]{Remark}

\newtheorem{proposition}[claim]{Proposition}
\newenvironment{proof}[1][Proof]{\textsl{#1.} }{\ \rule{0.4em}{0.7em}}

\begin{document}

\title{Hiatus perturbation for a singular Schr\"odinger operator with
an interaction supported by a curve in $\mathbb{R}^3$}
\author{P.~Exner$^{a,b}$ and S.~Kondej$^c$}
\date{}
\maketitle
\begin{quote}
{\small \em a) Nuclear Physics Institute, Academy of Sciences,
25068 \v Re\v z \\ \phantom{a) }near Prague, Czech Republic
\\
b) Doppler Institute, Czech Technical University,
B\v{r}ehov{\'a}~7, \\ \phantom{a) }11519 Prague, Czech Republic
\\
c) Institute of Physics, University of Zielona G\'{o}ra, ul.
Szafrana 4a, \\ \phantom{a) } 65246 Zielona G\'{o}ra, Poland
\\
\phantom{a) }\texttt{exner@ujf.cas.cz},
\texttt{skondej@proton.if.uz.zgora.pl} }
\end{quote}

\begin{quote}
{\small We consider Schr\"odinger operators in $L^2(\R^3)$ with a
singular interaction supported by a finite curve $\Gamma$. We
present a proper definition of the operators and study their
properties, in particular, we show that the discrete spectrum can
be empty if $\Gamma$ is short enough. If it is not the case, we
investigate properties of the eigenvalues in the situation when
the curve has a hiatus of length $2\epsilon$. We derive an
asymptotic expansion with the leading term which a multiple of
$\epsilon \ln\epsilon$.}
\end{quote}

\section{Introduction}\label{introduction}
Singular Schr\"odinger operators with interactions supported by
manifolds of a lower dimension are not a new topic; their
properties were investigated already in the beginning of the
nineties \cite{BT} or even earlier in cases of a particular
symmetry, see, e.g., \cite{AGS, Sha}. Recently, a new motivation
appeared when people realized that such operators with
\emph{attractive} interaction provide us with a model of ``leaky''
quantum graphs which have the nice properties of graph description
of various nanostructures --- see, e.g., the proceedings volumes
\cite{BCFK}, \cite{EKST}, and references therein --- but they are
more realistic taking possible tunneling between the involved
quantum wires into account.

A series of papers devoted to this problem started by \cite{EI}
and we refer to \cite{Ex} for a bibliography; among the questions
addressed were geometrically induced spectral properties
\cite{EI}, scattering \cite{EK2}, approximations by point
interaction Hamiltonians \cite{EN, BO} or strong-coupling
asymptotic behavior \cite{EY1}. Another result concerns
perturbations of such Hamiltonians caused by alterations of the
interaction support. In \cite{EY2} the asymptotic behavior for the
eigenvalue shift was derived in the situation when the support is
a manifold of codimension one with a ``hole'' which shrinks to
zero, in particular, a curve in $\mathbb{R}^2$ with a hiatus; it
was shown that in the leading order the perturbation acts as a
repulsive $\delta$  interaction with the coupling strength
proportional to the hole measure (in particular, the hiatus
length).

The aim of this paper is to analyze the analogous question in the
situation where the codimension of the manifold is two,
specifically, for an interaction supported by a curve in
$\mathbb{R}^3$. The extension is by far not trivial since the
codimension of the singular interaction support influences
properties of such Schr\"odinger operators substantially
\cite{AGHH}. In our particular case we know that to define such a
Hamiltonian for a curve in $\mathbb{R}^3$ one cannot use, in
contrast to the codimension one case, the ``natural'' quadratic
form and has to resort to appropriate generalized boundary
conditions \cite{EK1}.

We are going to demonstrate that the asymptotic behavior of the
eigenvalues with respect to the hiatus length $\epsilon$ is of the
following form,
$$
\lambda _j (\epsilon )=\lambda_{L} -s_j (\lambda_L )\epsilon \ln \epsilon
+o(\epsilon \ln \epsilon )\,,\quad j=m,...,n\,,
$$
where $\lambda_L $ is an unperturbed eigenvalue of the Hamiltonian
corresponding to the absence of the hiatus, the indices run
through a basis in the corresponding eigenspace so that $n-m+1$ is
the multiplicity of $\lambda_L$, and $s_j(\lambda_L )$ are
coefficients specified in Theorem~\ref{th-final}. This shows that
the asymptotics is in the case of codimension two is substantially
different --- recall that for codimension one the second term is
linear in $\epsilon$ --- due to the more singular interaction
involved. The dependence on the codimension is manifested also in
other ways. For instance, while a nontrivial and attractive
interaction supported by \emph{any} manifold of codimension one
gives rise to bound states, in the situation discussed here a
minimum curve length is needed to produce binding as we will
demonstrate in Section~\ref{boundsta}.

\setcounter{equation}{0}
\section{Preliminaries }\label{preliminaries}
As mentioned in the introduction, we are interested in generalized
Schr\"odinger operators with singular potentials supported by
sets of lower dimensions. In our case the support of the singular
potential will be a finite $C^1$ smooth curve in $\R^3$ of length
$L$ without self-intersections which may and may not be a loop;
the corresponding Schr\"odinger
operator can be formally written as
\begin{equation} \label{formaldelta}
-\Delta -\tilde\alpha \delta (x-\Gamma ) \quad \mathrm{with}
\quad\tilde\alpha <0 \,.
\end{equation}
We mark the parameter in this formal expression by a tilde to
stress that it is different from the true ``coupling constant''
which will introduce below. It is a natural requirement that the
operator which gives a mathematical meaning to (\ref{formaldelta})
should act as the Laplacian on the domain $C^\infty _0 (\R^3
\setminus \Gamma )$, which motivates us to look for self-adjoint
extensions of the symmetric operator $-\dot{\Delta} : C^\infty _0
(\R^3 \setminus \Gamma ) \mapsto L^2:= L^2 (\R^3) $ such that
$\dot{\Delta}f=\Delta f $. The deficiency indices of
$-\dot{\Delta} $ are infinite, of course, and looking for
operators giving a meaning to (\ref{formaldelta}) we will restrict
ourselves to  a certain ``local'' one-parameter family of
extensions which will be specified in the next section.

We have to say also something more about $\Gamma $ and to
introduce a family of auxiliary ``comparison'' curves in its
vicinity. Since $\Gamma$ is $C^1$ smooth by assumption it admits a
parameterization by the arc length. This means that $\Gamma $ is a
graph of a $C^1$ function $\gamma :[0,L ] \mapsto \R^3$ such that
$|\dot\gamma(s)|=1$, where $\dot{\gamma}$ stands for the
derivative. Moreover, we assume that
\begin{description}
\item{$\mathrm{(a)}$} there exist $c>0$ and $\mu >1$ such that
$$
|\gamma (s)-\gamma (t)|\geq |s-t |(1-c|s-t |^\mu )\quad
\mathrm{for }\quad c|s-t|^\mu <1\,.
$$
\end{description}
If $\Gamma $ is not a closed curve then, of course, one of its
endpoints is given by $\gamma(0)$. However, if $\Gamma $ is a loop
then there is no such natural ``starting point'' and we assume
that the above property is valid independently of the way the loop
is parametrized.

We will say that a family of curves $\{\Gamma_d\}$ is neighboring
with $\Gamma$ if they are graphs of functions $\gamma_d :
[0,L]\mapsto \R^3$ with the following properties for any $s\in
[0,L]$ and $d$ small enough
\begin{description}
 \item{$\mathrm{(b1)}$} $\;|\gamma (s)-\gamma _d (s)|=d\,$,
 \item{$\mathrm{(b2)}$} $\;|\dot{\gamma} (s)-\dot{\gamma}_d
(s)|=\mathcal{O}(d)\,$ as $d\to 0$,
 \item{$\mathrm{(b3)}$} $\;\gamma (s)-\gamma_d (s)$ is
 perpendicular to $\mathrm{t}_d(s):=\dot{\gamma}_d (s)\,$;
\end{description}
 the error term is assumed to be uniform on $[0,L]$.
For instance, if the Fr\'enet frame $(\mathrm{t}, \mathrm{n},
\mathrm{b})$  is defined globally for $\Gamma $ then any family of
``shifted'' curve defined as the graphs of
$$
\gamma + \eta_1 \mathrm{n}+ \eta_2 \mathrm{b} : [0,L]\mapsto
\R^3\,, \quad |\eta| =\sqrt{\eta_1^2+\eta_2^2}=d\,,
$$
is neighboring with $\Gamma$.

\setcounter{equation}{0}
\section{Definition of Hamiltonian and its resolvent}\label{resolvent}
In this section we shall construct an operator corresponding to
the formal expression (\ref{formaldelta}). As mentioned above, it
will be defined as a self-adjoint extension of $-\dot{\Delta} $.
To this aim we will follow the scheme proposed by Posilicano
\cite{Po01, Po04} which generalizes the standard Krein's theory.
The self-adjoint extensions are parametrized in it by
Birman-Schwinger-type operators entering into expression of their
resolvents. As usual, such a resolvent consists of a `free' term
and a `perturbative' remainder. Since we are in three dimensions,
the `free' resolvent is given by $R_z =(-\Delta -z )^{-1}:L^2
\mapsto L^2$, $z\in \rho (-\Delta )$, which is an integral
operator with the kernel
\begin{equation}\label{kernel_R}
G_z (x,y)=\frac{1}{4\pi }\frac{\e^{-{\sqrt{-z}}|x-y|}}{|x-y|}\,.
\end{equation}
As there is no risk of confusion we will use the same notation
$G_z (\cdot)$ for the function of a scalar argument, i.e. $G_z
(\rho )=\e^{-{\sqrt{-z}}|\rho|} (4\pi|\rho|)^{-1} $, $\rho \in \R
\setminus \{0\}$.

To construct the second term of the resolvent we need an embedding
to the Hilbert space associated with the support of our singular
potential. Such a space is naturally defined by $L^2 (\R^3 ,
\mu_\Gamma)$, where $\mu_\Gamma $ denotes the Dirac measure on
$\Gamma $. It is convenient, however, to use a natural
identification $L^2 (\R^3 , \mu_\Gamma ) \simeq L^2 (I)$, $I\equiv
(0,L)$ which we will do in the following. Is is well known that
$R_z$ defines a unitary map between $L^2$ and $W^{2,2}(\R^3)\equiv
W^{2,2}$. Moreover, with reference to the Sobolev theorem we claim
that the trace operator $\tau : W^{2,2}\to L^{2}(I)$ is
continuous, and consequently, the following operators,
$$
\mathrm{R}_{z}:=\tau R_{z}:L^{2}\to L^{2}(I)\,,\quad
\mathrm{R}_{z}^{\ast}: L^{2}(I)\to L^{2}\,,
$$
where $\mathrm{R}_{z}^{\ast}$ is the adjoint to $\mathrm{R}_{z}$, are
continuous as well.
\subsection{Birman--Schwinger operator}
The mentioned Birman-Schwinger operators are defined as symmetric
operators $\Theta_z$ in $L^2 (I)$ parameterized by $z\in \rho
(-\Delta )$ and satisfying the pseudo-resolvent equivalence
\begin{equation}\label{eq-peudoresolvent}
\Theta _w -\Theta _z =(w-z)\mathrm{R}_w\mathrm{R}^{\ast}_z \quad
\mathrm{for}\quad w\,, z \in \rho (-\Delta )\,.
\end{equation}
Furthermore, if the set $Z := \{z\in \rho (-\Delta ) : \Theta
^{-1}_{z}\, \mathrm{exists\,\, and \,\, is \,\,bounded }\}$ is
nonempty then the following operator
\begin{equation}\label{eq-resolvent}
R_{z; \alpha }= R_z -
\mathrm{R}_{z}(\Theta_z)^{-1}\mathrm{R}_{z}^{\ast}\,
\quad\mathrm{for}\quad z\in Z
\end{equation}
defines the resolvent of certain self-adjoint extension of
$-\dot{\Delta} $, cf.~\cite{Po01}. Our aim is now to find such an
operator $\Theta _z$ satisfying (\ref{eq-peudoresolvent}) and
corresponding to the singular potential defined by a certain
coupling constant. The explicit form of such operator was
discussed by \cite{BT} but for our purpose it is useful to derive
the other, albeit equivalent form of $\Theta_z$ (recall that in
the mentioned paper the potential was defined more generally, as a
function on $I$; in our model it is just a ``coupling'' constant
which we will denote as $\alpha$.)

The most natural way of determining $\Theta_z$ would be to take
the embedding of $\mathrm{R}_{z}$ to $L^2 (I)$, as it is done on
the codimension one case \cite{BEKS}. However, the explicit
formula for $G_z$, the kernel of $\mathrm{R}_{z}$, shows that the
expression $\tau \mathrm {R}^{\ast}_{z}$ does not make sense
because $G_z$ has a singularity; to make use of the approach
sketched above the singularity has to be removed by an appropriate
regularization. To put it differently, the operator $\tau$ can not
be canonically extended onto $L^2$ which is the range of $\mathrm
{R}^{\ast}_{z}$. On the other hand, to preserve the equivalence
(\ref{eq-peudoresolvent}) we have to consider a special type of
regularization which does not depend on $z$. Using standard facts
from the Sobolev space theory \cite{RS} we claim that $\mathrm
{R}^{\ast}_{z}f\in W^{2,2}_{\mathrm{loc}}(\R^3 \setminus \Gamma
)\cap C^{\infty } (\R^3 \setminus \Gamma )$ for $f\in L^2 (I)$,
thus the embedding $\mathrm{R}^{\ast}_{z}f\upharpoonright_{\Gamma
_d}$ is a $C^\infty $ function on $I$; recall that $\Gamma_d$ was
introduced in Section~\ref{preliminaries} as a neighbooring curve
with $\Gamma$. With these facts in mind we introduce a logarithmic
regularization defined through the pointwise limits
\begin{equation}\label{limiregularization}
\breve{f}(s)=\lim_{d\to 0}\: \Big[\mathrm
{R}^{\ast}_{z}f\upharpoonright
 _{\Gamma _d} (s)+\frac{1}{2\pi} f (s)\ln d\, \Big]\,\quad \mathrm{for
}\quad s\in I\,.
\end{equation}
By virtue of the following lemma, the relation
(\ref{limiregularization}) defines a function belonging to $L^2
(I)$ provided $f\in W^{1,2} (I)$.
\begin{lemma} \label{lemmaQ}
The operator $Q_z$ defined by the relation
\begin{equation}
\nonumber (Q_{z}f)(s):= \frac{1}{4\pi}\left(
\int_{I}\frac{f(t)-f(s)}{|t-s|}\,\mathrm{d}t +f(s)\ln 4 s(L\!-\!s)
\right) + \int_{I}\mathcal{R}_z(s,t) f(t)\mathrm{d}t\,,
\label{formQ}
\end{equation}
where
\begin{equation}\label{eq-formR}
\mathcal{R}_z (s,t):= G_{z}(\gamma (s)-\gamma
(t))-(4\pi|s-t|)^{-1}\,,
\end{equation}
maps $W^{2,1}(I)\mapsto L^2 (I)$ and $Q_z f= \breve{f}$.
\end{lemma}
In the following we will also employ the decomposition of the kernel
of the last term in (\ref{formQ}), namely $\mathcal{R}_z (s,t) =
\mathcal{A}_{z}(|s-t|)+\mathcal{D}_{z}(s,t)$, where
\begin{equation}\label{eq-reprAD}
\mathcal{A}_{z}(\rho ):=G_z (\rho )-(4\pi |\rho | )^{-1}\,,\quad
\mathcal{D}_{z}(s,t):= G_{z}(\gamma (s)-\gamma (t))-G_{z}(s-t)\,.
\end{equation}
Before starting the proof of the lemma let us make a couple of
comments. The operator $Q_z$ defined on the space $W^{1,2} (I)$
with the topology inherited from $L^2 (I)$ is essentially
self-adjoint. Taking its closure we obtain its unique self-adjoint
extension in $L^2 (I)$ for which we will use the same notation.

Furthermore, $Q_z$ satisfies the relation
(\ref{eq-peudoresolvent}). Indeed, let us note first that the
first resolvent formula for $R_z$ can be extended by the
continuity to the equivalence $\mathrm{R}^\ast_w
-\mathrm{R}^\ast_z =(w-z)R_w \mathrm{R}^\ast_z$. Since $R_w
\mathrm {R}^{\ast}_{z}f$ is a continuous function as an element of
$W^{2,2}$ we can take the limit
$$
\lim_{d\to 0} (\mathrm {R}^{\ast}_{w}-\mathrm {R}^{\ast}_{z})
f\upharpoonright _{\Gamma _d}= (w-z)R_w \mathrm {R}^{\ast}_{z}f
\upharpoonright _{\Gamma _d}\,,\quad f\in W^{1,2}(I)\,,\;\, w,z
\in \rho (-\Delta )\,,
$$
which consequently, in view of (\ref{limiregularization}), gives
\begin{equation}\label{eq-peudoresolvent2}
 (Q_w - Q_z)f = (w-z) \mathrm {R}_w\mathrm
{R}^{\ast}_{z}f\,.
\end{equation}

 \medskip

\begin{proof}[Proof of Lemma~\ref{lemmaQ}]
We break the argument into two parts: \\ \emph{Step 1:} Assume
first that $\Gamma$ is a straight line segment, i.e. we have
$\mathcal{D}_{z}=0$. Let us decompose $\mathrm
{R}^{\ast}_{z}f\upharpoonright _{\Gamma _d}$ into the sum of two
terms,
\begin{equation} \label{decomposition0}
\nonumber \mathrm {R}^{\ast}_{z}f\upharpoonright _{\Gamma
_d}(s)=\int_{I}G^{d}_{z}(s\!-\!t)f(t)\,\mathrm{d}t=
\int_{I}\mathcal{S}^{d}(s\!-\!t)f(t)\,\mathrm{d}t+
\int_{I}\mathcal{A}^{d}_{z}(s\!-\!t)f(t)\,\mathrm{d}t\,,
\end{equation}
where $ G^{d}_{z}(\rho ):= G_{z}((d^{2}+\rho^{2})^{1/2})$ and
$$
S^{d}(\rho ):= (4\pi (d^{2}+\rho^{2})^{1/2})^{-1}  \,,\quad
\mathcal{A}^{d}_{z}(\rho):=
G_{z}(\rho)-\mathcal{S}^{d}(\rho)\,.$$
The first term at the r.h.s of (\ref{decomposition0}) can be
rewritten as follows
\begin{equation}\label{decomposition1}
\int_{I}\mathcal{S}^{d}(s-t)f(t)\mathrm{d}t=
\int_{I}(f(t)-f(s))\mathcal{S}^{d}(s-t)\mathrm{d}t+
f(s)\int_{I}\mathcal{S}^{d}(s-t)\mathrm{d}t\,.
\end{equation}
The integrated function in the first term at the r.h.s. of the
last relation can be bounded by $|f(t)-f(s)|(4\pi |t-s|)^{-1}$
which belongs to $L^2 (I)$ in view of the fact that $f\in
W^{1,2}(I)$. Hence employing Lebesque's dominated convergence
theorem we can conclude that the first term at the r.h.s. of
(\ref{decomposition1}) tends to $\int_{I}(f(t)-f(s))(4\pi
|t-s|)^{-1}\,\mathrm{d}t$ for $d\to 0$. To handle the second term
we decompose it integrating separately along $I_{\delta}=I_{\delta
, s }:=\{t \in I :|t-s|<\delta \}$ and $I^c _\delta :=I\backslash
I_\delta $ for $\delta$ small enough. As a result we arrive at
\begin{equation}\label{Idelta}
\int_{I_\delta }\mathcal{S}^{d}(s-t)\,\mathrm{d}t
=\frac{1}{2\pi}\,\big[ \ln \omega_d (\delta )  -\ln d
\,\big]\,,\quad \omega_d (a):= a + (a^ 2+d^2)^{1/2} \,,
\end{equation}
\begin{equation} \label{-Idelta}
\nonumber \int_{I^c _\delta }\mathcal{S}^{d}(s-t)\,\mathrm{d}t=
\frac{1}{4\pi} \ln \omega _d (s)\, \omega _d (L-s)
 -\frac{1}{2\pi}\ln
\omega_d (\delta )\,.
\end{equation}
Combining the above results with (\ref{decomposition1}) we get
\begin{eqnarray}
\lefteqn{\lim_{d\to 0}\left(
\int_{I}\mathcal{S}^{d}(s-t)f(t)\,\mathrm{d}t +\frac{1}{2\pi}
f(s)\ln d \right)} \nonumber \\ && = \frac{1}{4\pi } \left(
\int_{I}\frac{f(t)-f(s)}{|t-s|}\,\mathrm{d}t+ f(s)\ln 4s(L-s)
\right)\,. \label{e-estimateS}
\end{eqnarray}
To obtain the result we have to handle the limit $\int_I
\mathcal{A}^{d}_{z}(s,t) f(t) $ as $d\to 0$, cf.~
(\ref{decomposition0}). Using the dominated convergence again we
find
\begin{equation}\label{e-estimateA}
\lim _{d\to 0}\int_I \mathcal{A}^{d}_{z}(s,t) f(t) \mathrm{d}t =
\int_I \mathcal{A}_{z}(s,t) f(t) \mathrm{d}t\,,
\end{equation}
which reproduces the remaining term at the r.h.s. of
(\ref{formQ}). Putting together (\ref{e-estimateS}) and
(\ref{e-estimateA}) we arrive at the sought result,
\begin{eqnarray*} \lefteqn{\lim_{d\to 0}\left(
\mathrm {R}^{\ast}_{z}f\upharpoonright _{\Gamma
_d}(s)+\frac{1}{2\pi} f(s)\ln d \right)}  \\ && =
\frac{1}{4\pi}\left( \int_{I}\frac{f(t)-f(s)}{|t-s|}\,\mathrm{d}t
+f(s)\ln 4 s(L-s) \right) + \int_{I}\mathcal{A}_z(s,t)
f(t)\mathrm{d}t\,.
\end{eqnarray*}
\emph{Step 2:} Consider next the general case when $\Gamma$ is a
finite curve which may and may not be closed. Then we employ the
decomposition
\begin{eqnarray*}
\lefteqn{\lim_{d\to 0}\left( \mathrm
{R}^{\ast}_{z}f\upharpoonright _{\Gamma _d}(s)\right) =
\int_{I}G_{z}(\gamma_{d}(s)-\gamma(t))f(t)\,\mathrm{d}t} \\ && =
\int_{I}G_{z}^{d}(s-t)f(t)\,\mathrm{d}t+
\int_{I}\mathcal{D}_{z}^{d}(s,t)f(t)\,\mathrm{d}t\,,
\end{eqnarray*}
where $\mathcal{D}_{z}^{d}(s,t) :=G_{z}(\gamma_{d}(s)
-\gamma(t))-G_{z}^{d}(s-t)$ and $\gamma_{d}$ is the function whose
graph is the neighbooring curve $\Gamma_d$ with $\Gamma$.
Using the result of the first step and the limit
$$
\lim_{d\to 0}\int_{I}\mathcal{D}_{z}^{d}(s,t)f(t)\,\mathrm{d}t=
\int_{I}\mathcal{D}_{z}(s,t)f(t)\,\mathrm{d}t
$$
discussed in~Remark~\ref{re-LDCTh} below we get the claim,
concluding thus the proof of the lemma.
\end{proof}

\medskip

Now we can express the \emph{resolvent of the Hamiltonian}. As was
already mentioned one can parameterize self-adjoint extensions of
certain symmetric operators by means of operators satisfying the
pseudo-resolvent formula. In our model we are specifically
interested in extensions of $-\dot{\Delta} : C^\infty _0 (\R^3
\setminus \Gamma ) \mapsto L^2:= L^2 (\R^3) $. The operators
$\Theta_z= Q_z -\alpha $ are suitable candidates for the role of
Birman--Schwinger operators because they satisfy the relation
(\ref{eq-peudoresolvent}). The parameter $\alpha \in \R$ appearing
here will be referred to as the coupling constant. It is certainly
different from the $\tilde\alpha$ appearing in
(\ref{formaldelta}); it is enough to notice that the absence of
the interaction is associated with the value $\alpha=\infty$.

To complete the argument one has to make an \emph{a posteriori}
claim that the set $Z$ is nonempty, which will be done in
Section~\ref{boundsta} below. Summing up the discussion, the
operator
\begin{equation}\label{e-resolvent}
R_{z; \alpha }= R_z - \mathrm{R}_{z}(Q_z -\alpha
)^{-1}\mathrm{R}_{z}^{\ast}\,\quad\mathrm{for}\quad z\in Z
\end{equation}
is in view of the mentioned result in \cite{Po01} the resolvent of
a self-adjoint extension of $-\dot{\Delta}$. We will regard it as
a rigorous counterpart of the formal Hamiltonian
(\ref{formaldelta}) and denote it in the following as $\Hag$.

\subsection{Alternative forms of $Q_z$}
The need to introduce a renormalization makes the use of
Birman--Schwinger approach more complicated than in the
codimension one case. In addition, the way we have chosen above, with
the limit taken over a family of comparison curves ``parallel'' to
the entire $\Gamma$ is not a particularly elegant one. It is
possible to think of other regularizations defining $Q_z$  by
$$\int_I
G_z (\tilde\gamma_d(s)-\gamma(t))f(t)\,\mathrm{d}t +\frac{1}{2\pi}
f(s)\ln d \,,
$$
where $\tilde\gamma_d$ correspond to another curve family. One
possibility is to consider curves which coincide with $\Gamma$
everywhere except in the vicinity of the singularity, the point
with $s=t$, where they have a recess the size of which is
controlled by the parameter $d$. To describe this and other
possible regularizations we will look at a more general class into
which all of them fit.

Given $s$ we consider a family of $C^2$ curves $\tilde{\Gamma }_{d,s}  $
which are graphs of
$\tilde{\gamma }_{d,s} \equiv \tilde{\gamma }_d  :[0,L] \mapsto \R^3 $
with $d:=|\tilde{\gamma }_d (s)- \gamma_d (s)| =\|\tilde{\gamma }_d
-\gamma\|_\infty$. The assumptions (b) of
Section~\ref{preliminaries} will be then replaced by the following
modified ones; for any $t\in [0,L]$ and $d$ small enough,
\begin{description}
 \item{$\mathrm{(\widetilde b1)}$}
 $\;|\gamma (t)-\tilde\gamma_d (t)|=\mathcal{O}(d)\,$ as $d\to 0$\,,
 \item{$\mathrm{(\widetilde b2)}$} $\;|\dot{\gamma} (t)
 -\dot{\tilde{\gamma}}_d
(t )|=\mathcal{O}(d)\,$ as $d\to 0$\,,
 \item{$\mathrm{(\widetilde b3)}$} $\;\gamma (s)
 -\tilde\gamma_d (s)$ is perpendicular to $
 \tilde\mathrm{t}_d(s):=
\dot{\tilde{\gamma }}_d (s)\,$,
\end{description}
where we suppose also that the error terms are uniform on $[0,L]$.
Let us stress that, in distinction to $\Gamma _d$, the curve
$\Gamma _{d,s}$ is in general \emph{not} parameterized by its arc
length. Then we have the following theorem the proof of which we
postpone to Section~\ref{appendix}.
\begin{theorem} \label{th-ear}
Under the stated assumptions,
\begin{equation}\label{eq-ear}
(Q_z f)(s)=\lim_{d\to 0 }\left[ \int_I G_z (\gamma
(s)-\tilde{\gamma }_d (t)) j_d (t)f(t)\,{d}t+ \frac{1}{2\pi}
f(s)\ln d \right]\,,
\end{equation}
where $j_d (s):= \left(\sum_{i=1} ^{3}(\dot{\tilde{\gamma }}_{d,i}
(s)^2 \right)^{1/2}$.
\end{theorem}

\setcounter{equation}{0}
\section{Existence of bound states }\label{boundsta}
We have said that in distinction to the codimension one case an
attractive interaction supported by a finite curve may not induce
bound states. The aim of this section is to make this claim
precise and to find conditions under which the Hamiltonian $\Hag$
has a nonempty discrete spectrum. Since the singular potential in
our model is supported by a compact set it is easy to check the
stability of the essential spectrum,
$$
\sigma _{\mathrm{ess}}(\Hag )=\sigma
_{\mathrm{ess}}(-\Delta)=[0,\infty)\,,
$$
see~\cite{BT}. This means that the negative halfline can contain
only the discrete spectrum $\sigma _{\mathrm{d}} (\Hag )$, and
consequently the set $Z$ of (\ref{eq-resolvent}) is nonempty.
Looking for negative eigenvalues we put $z=\lambda$ with $\lambda
<0$. We employ the Birman--Schwinger philosophy, specifically the
following result,
\begin{equation}\label{existenceBS}
\lambda\in \sigma_{\mathrm{d}}(\Hag)\,\Leftrightarrow \,
\mathrm{ker}(Q_{\lambda}-\alpha )\neq \{ 0 \} \,,
\end{equation}
where the multiplicity of $\lambda$ is equal to
$\mathrm{dim}\,\mathrm{ker}(Q_{\lambda}-\alpha )$ --
cf.~\cite{Po04}. In addition, the eigenfunctions of $\Hag$
corresponding to an eigenvalue $\lambda$ are given by
\begin{equation}\label{eq-eigenfunction}
\psi_\lambda =\mathrm{R}^\ast_\lambda  \phi_\lambda \,,\quad
\mathrm{where }\quad \phi_\lambda \in \ker (Q_{\lambda }-\alpha
)\,.
\end{equation}

It is well-known that a point interaction in $\R^2$ always
attractive, i.e. it gives rise for any $\alpha\in\R$ to exactly
one bound state with the eigenvalue
\begin{equation}\label{e-eigenvalue1}
\xi_0 = \xi_0(\alpha) = -4 \mathrm{e}^{2(-2\pi \alpha +\psi
(1))}\,,
\end{equation}
where $\psi (1)=-0,577...$ is Euler--Mascheroni constant,
cf.~\cite{AGHH}. Asking about existence of bound states in our
model, one may naively expect the same behavior as the
perturbation is again of codimension two. It appears, however,
that it is not so due to the presence of the third dimension which
makes a finite curve in $\R^3$ ``more singular'' than a point in
$\R^2$. We will show that if the length of curve is small enough
then our system has no bound states. We need the following result
auxiliary result.

\begin{lemma} \label{positivity-groud}
Suppose that $\sigma _{\mathrm{d}} (\Hag) \neq \emptyset$. Then
the ground-state eigenvalue $\lambda_0 =\min \{\lambda \in \sigma
_{\mathrm{d}} (\Hag)\}$ is simple and the corresponding
eigenfunction $\psi_0 :=\psi_{\lambda_0}$ is a multiple of a
positive function.
\end{lemma}
\begin{proof}
We will employ the form associated with $-Q_z +\alpha $,
cf.~\cite{BT},
\begin{equation}
\varsigma _z [f] =-\frac{1}{2}\int_{I\times I}|f(s)-f(t)|^2
\,G_{z}(\gamma (s)-\gamma (t))\,\mathrm{d}t \mathrm{d}s
-\int_{I}|f(s)|^2 (a_z (s)+\alpha )\,\mathrm{d}s\,,
\end{equation}
where, with the notation introduced above,
\begin{eqnarray*}
\lefteqn{a_z (s) :=-\int_{I_{\delta }} G_{z}(\gamma (s)-\gamma
(t))\,\mathrm{d}t} \\ && + \int_{I^c_{\delta }} \left(
\frac{1}{4\pi |s-t|}-G_{z}(\gamma (s)-\gamma (t)) \right)
\,\mathrm{d}t -\frac{1}{2\pi }\log 2\delta \,. \end{eqnarray*}
Let us note that the inequality $||f(s)|-|f(t)|| \leq |f(s)-f(t)|$
implies
$$
\varsigma _z [|f|] \leq \varsigma _z[f]\,.
$$
For a fixed $z < \inf \sigma (\Hag)$ the form $\varsigma _z$ is
strictly positive. Using the Beurling--Deny criterion together
with the other results from \cite[vol.~II, p.~204]{RS} we find
that $(-Q_z +\alpha )^{-1}$ is positivity preserving. Moreover,
the operators $R_{z}$, $\mathrm{R}_{z} $, and
$\mathrm{R}_{z}^{\ast}$ are positivity improving because they are
defined by means of the kernel which is strictly positive. Hence
referring to (\ref{e-resolvent}) we conclude that the resolvent
$R_{z; \alpha}$ of $\Hag$ is positivity improving, and using
\cite{RS} again we get the sought positivity of $\psi_0$.
\end{proof}

\bigskip

We begin the discussion concerning the existence of bound state by
analyzing the simplest case, namely the situation when $\Gamma$ is
a line segment.
\begin{lemma} \label{estimations}
Suppose that $\Gamma$ is a finite line segment of length $L$. If
$L<2\e^{2\pi \alpha}$ then there $\Hag$ has no bound states. On
the other hand, if $L>2\pi \e^{2\pi \alpha-\psi(1)}$ then there
exists at least one bound state.
\end{lemma}
\begin{proof} In order to prove the absence of bound states
under the condition $L<2\e^{2\pi \alpha}$ it suffices in view of
(\ref{existenceBS}) to show that
\begin{equation}\label{absence}
\sup \sigma (Q_{\lambda} )<\frac{1}{2\pi }\ln \frac{L}{2}\,.
\end{equation}
It is clear that the value $\sup \sigma (Q_{\lambda }) $ is
achieved by $(Q_{\lambda }\phi_{0},\phi_0)$, where $\phi_0\in \ker
(Q_{\lambda _0}-\alpha )$ is the normalized function corresponding
to the ground state $\psi_0$ by the relation $\psi_0
=\mathrm{R}^\ast _{\lambda_0 }\phi_0$,
cf.~(\ref{eq-eigenfunction}). Using the expression of $Q_{\lambda
}$ given by Lemma~\ref{lemmaQ} we get the following asymptotics,
\begin{equation}\label{asymptoticspsi}
\psi_0 \upharpoonright
 _{\Gamma _d} (s)=\mathrm {R}^{\ast}_{\lambda_0}\phi_0 \upharpoonright
 _{\Gamma _d} (s) \approx -\frac{1}{2\pi} \phi_0
(s)\ln d\, - (Q_{\lambda _0}\phi_0 )(s) \,\;\; \mathrm{as }\quad
d\to 0\,,\,\,s\in I\,.
\end{equation}
Since $\psi_0$ can be chosen positive by
Lemma~\ref{positivity-groud} and $\phi_0 \in W^{1,2}(I)$, as we
will demonstrate in Section~\ref{regularity} below, we come to the
conclusion that $\phi_0$ is positive as well because the leading
term of (\ref{asymptoticspsi}) is determined by $\phi_0$. Thus to
estimate $\sup \sigma (Q_{\lambda} )$ it is sufficient to consider
the expression $(Q_{\lambda}f,f)$ for positive functions $f$ only.
Using the relation (\ref{formQ}) again we find
\begin{eqnarray}
\lefteqn{(Q_{\lambda}f,f )= \int_{I\times I}\xi
_{f}(s,t)\,\mathrm{d}s \mathrm{d}t +\int_{I\times
I}\mathcal{A}_{\lambda}(s-t)f(s)f(t)\,\mathrm{d}s\mathrm{d}t} \nonumber \\
&& \phantom{AAA}+(4\pi)^{-1}\int_{I}f(s)^2 \ln 4 s\,
\mathrm{d}s\,, \phantom{AAAAAAAAAAAAA} \label{formQ2}
\end{eqnarray}
where
$$
\xi _{f}(s,t):=\frac{(f(t)-f(s))f(s)}{4\pi|s-t|}\,.
$$
A straightforward calculation yields the estimate
$$
\xi _{f}(s,t)-\xi _{f}(t,s)=-\frac{(f(t)-f(s))^2}{4\pi|s-t|}\leq
0\,,
$$
which in turn leads to the following inequality
\begin{eqnarray*}
\lefteqn{ \int_{I\times I}\xi _{f}(s,t)\,\mathrm{d}s\mathrm{d}t =
\int_{I}\int_{s<t}\xi
_{f}(s,t)\,\mathrm{d}s\mathrm{d}t+\int_{I}\int_{s>t}\xi
_{f}(s,t)\,\mathrm{d}s\mathrm{d}t} \\ \phantom{AAA} && =
\int_{I}\int_{s<t} (\xi _{f}(s,t)+ \xi
_{f}(t,s))\,\mathrm{d}s\mathrm{d}t \leq 0\,.
\phantom{AAAAAAAAAAAAAAA}
\end{eqnarray*}
On the other hand, we have $ \mathcal{A}_{\lambda}(\rho)\leq 0$
and the only positive contribution to $(Q_{\lambda}f,f)$ comes
from the last term of (\ref{formQ2}). Finally, the inequality
(\ref{absence}) is a consequence of
$$
\sup_{s\in I}\,(4\pi)^{-1}\ln s(L-s) =(2\pi )^{-1}\ln L/2\,.
$$
The second part of the lemma, the condition for the
\emph{existence} of bound states can be obtained by the Dirichlet
bracketing. To be more precise, assume that $\Gamma :=
\{(s,0,0)\,, s\in [0,L]\} $ and denote by $H_{\alpha , \Gamma }^D$
the Laplace operator with singular potential on $\Gamma$ and
Dirichlet boundary conditions at the planes $(0,y,z)$ and
$(L,y,z)$ with $y,z \in \R$; it is well known
\cite[Sec.~XIII.15]{RS} that
\begin{equation}\label{e-Dirichlet}
\inf \sigma (\Hag )\leq \inf \sigma ( H_{\alpha ,\Gamma }^D  )\,.
\end{equation}
Furthermore, using a simple separation of variables we find  $\inf
\sigma ( H_{\alpha ,\Gamma }^D )=-4\mathrm{e}^{2(-2\pi\alpha +\psi
(1) )}+L^{-2}$. The operator $H_{\alpha ,\Gamma }^D$ has always a
ground state with the eigenvalue which becomes negative for a
fixed $\alpha$ and $L$ large enough. Since the essential spectrum
of $\Hag$ is $\R^+$, by (\ref{e-Dirichlet}) and the minimax
principle the Hamiltonian $\Hag$ has at least one discrete
eigenvalue as well; working out the negativity condition
quantitatively we arrive at the conclusion.
\end{proof}
\begin{remark}
{\rm  The method we have used in the proof is not particularly
precise which explains the gap of $\pi\,e^{-\psi(1)}\approx 5.56$
in the ratio of the lengths $L$ for which the existence and
nonexistence of the discrete spectrum were established above.}
\end{remark}

Let us discuss next the general situation and consider a
nontrivial curve which again may or may not be closed. To be
concrete we consider a family of curves which are connected
subsets of a fixed $\Gamma$ corresponding to different
subintervals of the arc length parameter. The deviation of each
such curve from the corresponding straight segment is measured by
the quantity $\mathcal{D}_{\lambda}\neq 0$ given by
(\ref{eq-reprAD}). Since $|\gamma(s)-\gamma(t)|\leq |s-t|$ in view
of the used parameterization and the function $\rho\mapsto
\e^{-\rho}/ \rho$ is decreasing we find that
$\mathcal{D}_{\lambda}>0$ holds on an open set, and moreover
\begin{equation}\label{e-estD}
\mathcal{D}_{\lambda} (s,t)\leq \frac{1}{4\pi}\left(
\frac{1}{|\gamma (s)-\gamma (t)|}-\frac{1}{|s-t|} \right)\,.
\end{equation}
Using the assumption $(\mathrm{a})$ and mimicking the argument of
\cite{EK1} one can show that the operator with kernel defined by
the r.h.s. of (\ref{e-estD}) is bounded (or even Hilbert-Schmidt)
and denote its norm as $D$, (see also Remark~\ref{re-estR}).
Proceeding as in the proof of
Lemma~\ref{estimations} we arrive at the conclusion that the
operator $\Hag $ has no bound states if $L<2\e^{2\pi \alpha -D}$.
On the other hand, using arguments borrowed from \cite{EK1} we can
claim that in the case $L>2\pi \e^{2\pi \alpha-\psi(1)}$ the bound
states do not disappear when a segment is replaced by a curve of
the same length, since the bending acts as an effective attractive
interaction. Summarizing this discussion we have the following
result.
\begin{theorem} \label{th-exbound}
For a fixed $\alpha \in \R$ in the described situation, there
exists $L_\alpha>0 $ such that the operator $\Hag$ has no discrete
spectrum for $L<L_{\alpha}$. On the other hand, if $L>2\pi
\e^{2\pi \alpha-\psi(1)}$ then there is at least one bound state.
\end{theorem}

\setcounter{equation}{0}
\section{Regularity of eigenfunction} \label{regularity}
Before we proceed to our main result we need as a preliminary to
investigate the regularity of $\phi \in \ker (Q_{\lambda
_L}-\alpha )$, where $\lambda _L$ is an eigenvalue of $\Hag $;
specifically we will demonstrate that this function belongs to
$W^{1,2}$. The proof of this claim is involved and we divide it
into several steps. To simplify the presentation we will show
first the regularity of the corresponding eigenfunction in the
case when $\Gamma$ is a loop, and then we will comment on an
extension of the result.
The idea is to compare the loop with a circle of the same length.
Suppose $\Gamma $ is a closed curve satisfying the assumptions of
Section~\ref{preliminaries} and $\Gamma^c $ is a circle of the length
$L$; up to Euclidean transformations, $\Gamma^c$ is thus the graph
of the function $\gamma^c (\cdot ) =\frac{L}{2\pi} (\cos
\frac{2\pi}{ L} (\cdot ) ,\sin \frac{2\pi} {L}(\cdot ),0
):[0,L]\mapsto \R ^3 $. The operator $Q_z$ can be defined in
analogy with~(\ref{limiregularization}), i.e.
\begin{equation}\label{e-alternativeQ}
Q_z = T_z ^c + \mathcal{D}^c _z\,,
\end{equation}
where
\begin{equation}\label{limiregularization1}
T_z ^c f =\lim_{d\to 0}\, \Big[\,\mathrm
{R}^{\ast}_{z}f\upharpoonright _{\Gamma _d^c } +\frac{1}{2\pi}
f\ln d\,\Big]\,\quad \mathrm{for }\quad s\in (0,L)
\end{equation}
and $\mathcal{D}^c _z$ is given by the kernel $\mathcal{D}^c _z
(s,t):= G_z (\gamma (s)-\gamma (t))-G_z (\gamma ^c (s)-\gamma ^c
(t))$; in the above expression $\Gamma _d^c$ stands for a neighbooring
 curve with $\Gamma_c$ and the properties described in
Section~\ref{preliminaries}.
\begin{lemma}\label{le-regularityD}
Assume that the assumption $(\mathrm{a})$ is satisfied; then for
any function $f\in L^2 (I) $ we  have $ \mathcal{D}^c _z f\in
W^{1,2}(I)$.
\end{lemma}
The proof is quite technical and we postpone it to the appendix.
\begin{lemma}\label{regularityT-a}
For $\phi \in \ker (Q_{\lambda _L}-\alpha )$ we have $(T_z
^c-\alpha )\phi \in W^{1,2}(I)$.
\end{lemma}
\begin{proof}
Using the pseudo-resolvent formula (\ref{eq-peudoresolvent2}) for
$w=\lambda _L$ and the fact that $\mathrm {R}_w\mathrm
{R}^{\ast}_{z}\phi \in W^{1,2}(I)$ we  get $(Q_z -\alpha )\phi \in
W^{1,2}(I)$. Applying then the result of the previous lemma and
the decomposition (\ref{e-alternativeQ}) we get the claim.
\end{proof}

\medskip

This allows us finally to formulated the indicated result.
\begin{proposition}\label{le-reguarity}
Any eigenfunction $\phi \in \ker (Q_{\lambda _L}-\alpha )$ belongs
to $W^{1,2}(I )$.
\end{proposition}
\begin{proof}
Using the radial symmetry valid for $\Gamma ^c $ one finds
$$
T^c _z f =\sum _{k\in \mathbb{Z}} b_k (z)f_k\, \mathrm{e}^{i2\pi
k(\cdot)/L}\,,
$$
where $f_k$ are Fourier coefficients of $f$ and $b_k (z)\in
\mathbb{C}$. Hence $T^c _z$ commutes with the derivative operator
$D$, which implies for $z\in \mathbb{C}^+$
\begin{equation}\label{e-estDphi}
\|D\phi \|\leq C\|(T^c _z -\alpha )D \phi \|=C\|D(T^c _z -\alpha )
\phi \| < \infty\,, %
\end{equation}
where $C$ is a positive constant;  we have used the fact that $T^c
_z -\alpha $ is invertible with a bounded inverse in combination
with Lemma~\ref{regularityT-a}. The sought claim follows directly
from (\ref{e-estDphi}).
\end{proof}
\begin{remark}
{\rm In a similar way one can deal with the situation when the
curve $\Gamma$ is not closed; the idea is to compare it to a
circular segment. To be precise we introduce $\Gamma^c$ which is,
as before, a circle defined as the graph of $\gamma ^c :[0,L+d
]\mapsto \R^3 $, $d>0$ and its segment $\Gamma^{c,r}$ being the
graph of $\gamma ^{c,r} : [0,L ]\mapsto \R^3 $ such that $\gamma
^{c,r} (s)=\gamma ^{c}(s) $ for any $s\in [0,L]$. In analogy with
(\ref{e-alternativeQ}) we can decompose the operator $Q_z$
corresponding to $\Gamma$ as
$$ 
Q_z =T_z ^{c,r}+ \mathcal{D}_z ^{c,r}\,,
$$ 
where $T_z ^{c,r}$ and $\mathcal{D}_z ^{c,r}$ are defined as in
(\ref{e-alternativeQ}) but by means of $\gamma^ {c, r}$, with the
variable appropriately restricted. The proofs of
Lemmata~\ref{le-regularityD}, \ref{regularityT-a} can be mimicked
directly for the operators $T_z ^{c,r}$ and $\mathcal{D}_z
^{c,r}$. On the other hand, the proof of
Proposition~\ref{le-reguarity} requires some comments. Given
$\delta >0$ let us introduce the natural embeddings
$\breve{\mathcal{I}}: L^2 (0, L+\delta )\mapsto L^2 (0, L )$ and
$\breve{\mathcal{I}}^{\ast}: L^2 (0, L)\mapsto L^2 (0, L+\delta
)$. Using the explicit form of $Q_z$ given by (\ref{formQ}) we can
easily check that $ T_z ^{c,r}=\breve{\mathcal{I}} T_z ^c
\breve{\mathcal{I}}^{\ast} $. Now can repeat the reasoning which
leads to (\ref{e-estDphi}) but instead of the norm $\|\cdot \|$ in
$L^2 (0,L )$ we consider the norm $\|\cdot \|_\delta $ in $L^2
(\delta ,L-\delta )$, where $\delta>0$ is a constant which can be
made arbitrarily small; we get
\begin{eqnarray} \nonumber
\|D\phi \|_{\delta }\leq && C\|(T^{c,r} _z -\alpha )D \phi
\|_\delta =  C\|\breve{\mathcal{I}}(T^c _z -\alpha
)\breve{\mathcal{I}}^{\ast} D\phi \|_\delta  = \\&& \nonumber
C\|D\breve{\mathcal{O}}(T^c _z -\alpha )\breve{\mathcal{O}}^{\ast}
\phi \|_\delta < \infty\,.
\end{eqnarray}
This means that Proposition~\ref{le-reguarity} extends to the case
when $\Gamma $ is not a loop, by which the eigenfunction
regularity is finally established generally. }
\end{remark}

\setcounter{equation}{0}
\section{A curve with a hiatus}\label{Broken}
Now we finally come to our main topic. In this section we consider
the eigenvalue problem for a curve with a short hiatus. Suppose
that we have the system with the singular interaction supported by
a curve $\Gamma$ of length $L$ and satisfying the assumptions of
Section~\ref{preliminaries}. Naturally we have to exclude the
trivial case assuming that $\Hag$ has bound states; we know from
Theorem~\ref{th-exbound} a sufficient condition for that is
$L>2\pi \e^{2\pi \alpha- \psi(1)}$. For simplicity we will suppose
first that there is exactly one bound state with corresponding
eigenvalue $\lambda _L$; the generalization will be provided
at the end of this section.

Consider now a family of curves $\Gamma _\epsilon$ which coincides
with $\Gamma$ everywhere apart a short hiatus placed symmetrically
w.r.t $x_0=\Gamma (s_0)$, in other words, $\Gamma _\epsilon $ is a
graph of function $\gamma_\epsilon \,:\, [0,s_0 -\epsilon)\cup
(s_0 +\epsilon ,L] \mapsto \R^3$ and $\gamma_\epsilon (s)=\gamma
(s)$ for $s\in [0,s_0 -\epsilon)\cup (s_0 +\epsilon ,L]$. In the
following we will use the notations $I^c _\epsilon \equiv (0,s_0
-\epsilon)\cup (s_0 +\epsilon ,L)$ and $I_\epsilon$ for $(s_0
-\epsilon ,s_0 +\epsilon )$. Our aim is to derive asymptotics of
eigenvalue $\lambda (\epsilon )$ of $H_{\alpha ,\Gamma
_{\epsilon}}$ for $\epsilon $ small. Of course, we may expect that
$\lambda (\epsilon ) \to \lambda_L$ for $\epsilon \to 0$. Since,
as discussed above, the eigenvalue problem can be reduced in view
of \ref{existenceBS} to analysis of the Birman--Schwinger
operator, we will seek the function $\lambda (\epsilon)$ such that
$\mathrm{ker}(Q_{\lambda (\epsilon )} ^\epsilon-\alpha )$ is
nontrivial where $Q_{\lambda } ^\epsilon$ denotes the
Birman--Schwinger operator corresponding to $\Gamma _\epsilon$.
The first step towards that is to relate $Q_\lambda $ and
$Q_\lambda ^\epsilon$. It is convenient to introduce the natural
embedding maps acting between $L^2(I)$ and $L^2(I^c _\epsilon )$.
Let $\mathcal{I}_\epsilon$ stand for the canonical embedding from
$L^2(I)$ to $L^2(I^c _\epsilon )$ and $ \mathcal{I}^c _\epsilon$
for its adjoint acting from $L^2(I^c_\epsilon )$ to $L^2(I)$. We
will also use the abbreviation $\Qlec :=\mathcal{I}^c _\epsilon
\Qle \mathcal{I}_\epsilon$.
\begin{lemma}\label{th-estimate1}
The asymptotic expansion
\begin{equation}\label{1asmpt}
(Q_{\lambda }^\epsilon \mathcal{I}_\epsilon f ,
\mathcal{I}_\epsilon f)=(Q_{\lambda }f, f
)+\frac{2}{\pi}|f(s_0)|^2 \epsilon \ln \epsilon + o(\epsilon \ln
\epsilon )
\end{equation}
holds for $\epsilon \to \infty$ and any $f\in D(Q_\lambda )\cap
W^{1,2}(I)$.
\end{lemma}
\begin{proof}
Let us first note that for any $f\in L^2 (I)$ such that $\mathcal{I}_\epsilon
f\in D(\Qle )$ we have $(Q_{\lambda }^\epsilon
\mathcal{I}_\epsilon f , \mathcal{I}_\epsilon f)=(\Qlec f ,  f)$
and $\Qlec f $ can be decomposed as,
\begin{equation}\label{decompositionQ}
\Qlec f = \lim_{d\to 0}\left[ \int_{\bIe}G_\lambda (\gamma_d
(\cdot )-\gamma (t ))f(t)\mathrm{d}t+\frac{1}{2\pi }\ln d
\,f\right]\chi^c _\epsilon = Q_\lambda f-Jf -J^\prime f -Tf \,,
\end{equation}
where
\begin{eqnarray}
Jf &\!:=\!& \left[ \lim_{d\to 0}\int_{\Ie}G_\lambda (\gamma
_d(\cdot )-\gamma _d (t ))f(t)\mathrm{d}t \right]
\chi^c_\epsilon\,, \nonumber \\ [-1em] \label{definitionJ} \\
\quad J^{\prime}f &\!:=\!&\left[ \lim_{d\to 0}\int_{\bIe}G_\lambda
(\gamma_d ( \cdot ) -\gamma (t) )f(t)\mathrm{d}t \right]
\chi_\epsilon \nonumber
\end{eqnarray}
and
$$
Tf= \lim_{d\to 0}\left[ \int_{\Ie}G_\lambda (\gamma^d  (\cdot )-
\gamma (t) )f(t)\,\mathrm{d}t+\frac{1}{2\pi }\ln d
\,f\right]\chi_\epsilon\,.
$$
The symbols $\chi_\epsilon$, $\chi^c_\epsilon$ stand for the
characteristic functions of $I_\epsilon$ and $I^c _\epsilon$,
respectively. Let us show how the last term of
(\ref{decompositionQ}) emerges. In analogy with the proof of
Lemma~\ref{lemmaQ}, see eq.~(\ref{formQ}), one shows that
\begin{eqnarray}
\lefteqn{ (Tf)(s) = 
\frac{1}{4\pi} f(s)\ln 4(s-s_0 +\epsilon )(s_0-s +\epsilon
)\chi_\epsilon (s) } \nonumber \\ \label{1decompositionQ} && +
\left(\int_{I_\epsilon }\frac{f(t)-f(s)}{4\pi|s-t|}\,\mathrm{d}t
+\int_{I_\epsilon } \mathcal{R_{\lambda }}(s,t)f(t)\,\mathrm{d}t
\right) \chi_\epsilon (s)\,;
\end{eqnarray}
recall that $\mathcal{R _{\lambda }}(s,t)=\lim_{d\to 0} \mathcal{R
_{\lambda }}^d(s,t)=\lim_{d\to 0} (G_\lambda (\gamma_d (s)-\gamma
(t))-S^d (s-t))$ and $ S^d (s-t)=(4\pi(d^2 +
(s-t)^2)^{1/2})^{-1}$. Using the identity
$$
\int_{I_\epsilon}\ln 4(s-s_0+\epsilon )(s_0 -s+\epsilon
)\mathrm{d}s=8 \epsilon\ln 2\epsilon
$$
together with the expansion $f(s)=f(s_0)+o(1)$ for $s\sim s_0$,
which can be performed in view of the fact that $f\in W^{1,2}(I)$
we obtain
\begin{equation}\label{asT}
(Tf,f)=\frac{2}{\pi}|f(s_0)|^2 \epsilon \ln\epsilon +
\mathcal{O}(\epsilon)\,;
\end{equation}
note that the second and the third term of (\ref{1decompositionQ})
can be uniformly bounded w.r.t. $s$, cf. Remark~\ref{re-estR}
below, and consequently, they contribute in (\ref{asT}) to the
error term only. The latter depends on $\lambda $, however, it is
important for us that it can be uniformly bounded together with
its derivative being $\mathcal{O}(\epsilon)$.

Let us now consider the term $Jf $ appearing in
(\ref{decompositionQ}). Applying to (\ref{definitionJ}) the
decomposition $ G_\lambda (\gamma_d (s)-\gamma (t))= S^d (s-t)+
\mathcal{R _{\lambda }}^d(s,t)$ we get by a straightforward
computation
\begin{equation}\label{formJ}
(Jf)(s)=\left( (f(s_0)+o_\epsilon (1)) j_\epsilon (s)+\int
_{\Ie}\mathcal{R_{\lambda }}(s,t)f(t)\,\mathrm{d}t \right)
\chi^c_\epsilon (s)\,,
\end{equation}
where the error term $o_\epsilon (1)$ means the asymptotics for
$\epsilon \to 0$, and
$$
j_\epsilon (s):=\lim_{d\to 0}\int_{\Ie}S^d
(s-t)f(s)\,\mathrm{d}s=\frac{1}{4\pi} \ln \frac{|s-s_0|+\epsilon
}{|s-s_0|-\epsilon }\quad \mathrm{for}\quad |s-s_0|>\epsilon \,.
$$
Our aim is to estimate
\begin{equation}\label{e-scalarJ}
(Jf,f)=(f(s_0)+o_\epsilon (1))\int_{\bIe}j_\epsilon
(s)\overline{f(s)}\,\mathrm{d}s+\int_{\bIe}\int
_{\Ie}\mathcal{R_{\lambda }}(s,t)f(t)\overline{f(s)}\,\mathrm{d}t
\mathrm{d}s\,.
\end{equation}
By an analogous argument as in the first step of proof we can
check that the last term of (\ref{e-scalarJ}) contributes to
$\mathcal{O}(\epsilon)$. To handle the first term at the r.h.s. of
(\ref{e-scalarJ}) we integrate by parts
\begin{equation}\label{e-jbar}
\int_{I^c_\epsilon } j_\epsilon
(s)\overline{f(s)}\mathrm{d}s=\hat{j}_\epsilon
(s)\overline{f(s)}\mid _{\bIe} -\int_{\bIe }\hat{j}_\epsilon
(s)\overline{f'(s)}\mathrm{d}s\,,
\end{equation}
$$
\hat{j}_\epsilon (s):=\frac{1}{4\pi
}\,\sum_{k=\{-1,1\}}k(|s-s_0|-k\epsilon )\, \Big[\ln
(|s-s_0|-k\epsilon )-1 \Big]\, \frac{|s_0 -s|}{s_0 -s}\,.
$$
Consequently, the first term of (\ref{e-jbar}) takes the following form
$$
\hat{j}_\epsilon (s)\overline{f(s)}\mid _{\bIe}= -\frac{2}{\pi
}\epsilon\, \ln \epsilon \overline{f(s_0)}+o(\epsilon \ln \epsilon
)\,\quad \mathrm{for}\quad s \in \bIe\,.
$$
Furthermore, the second term can be estimated as
$$
\bigg|\int_{\bIe }\hat{j}_\epsilon
(s)\overline{f'(s)}\,\mathrm{d}s \bigg|\leq \|\hat{j}_\epsilon
\|_{L^2 (\bIe )}\|f\|_{W^{1,2}(I)}\,.
$$
One can check directly that $\|\hat{j}_\epsilon \|_{L^2 (\bIe )}=
o(\epsilon \ln \epsilon )$. Summarizing, we get
$$
\int_{I^c_\epsilon } j_\epsilon (s)\overline{f(s)}\,\mathrm{d}s=
-\frac{2}{\pi }\epsilon \ln \epsilon \overline{f(s_0)}+o(\epsilon
\ln \epsilon )\,, $$
and consequently, $(Jf,f)=-\frac{2}{\pi }|f(s_0)|^2\epsilon \ln
\epsilon +o(\epsilon \ln \epsilon )$. Using the fact that
$(Jf,f)=(J^{\prime}f,f)$ in combination with (\ref{asT}) we get
the claim.
\end{proof}

\bigskip

With the above lemma we are ready to demonstrate the following
result.
\begin{lemma} \label{le-existence.ev}
The eigenvalues of $\Qle$ tend to the eigenvalues of
$Q_\lambda $ for $\epsilon \to 0$. Moreover, if $\epsilon $ and
$\lambda - \lambda_L$ are small enough the
operator $\Qle$ has an eigenvalue $\eta (\lambda ,\epsilon )$
which tends to $\alpha $ as $\epsilon \to 0$ and $\lambda \to
\lambda_L$.
\end{lemma}
\begin{proof}
Since $\Qlec $ is the natural embedding of $\Qle$ to space $L^2
(I)$ it suffices to show the claim for $\Qlec$. Let us make the
following decomposition
\begin{equation}\label{e-Qlec}
(\Qlec f , f)=\left( (\Qlec f , f)-(Q_{\lambda }f, f
)\right)+\left((Q_{\lambda }f, f )-(Q_{\lambda_L }f, f )\right)+(
Q_{\lambda_L }f, f )\,.
\end{equation}
The convergence of the first term at the r.h.s. of (\ref{e-Qlec})
is proved in the pre\-vious lemma, precisely  we have $0<
(Q_{\lambda }f, f )-(\Qlec f , f)\to 0$ for $\epsilon \to 0$;
combining this with the results of \cite[Chap.~XIII]{Ka} we arrive
at the first statement of the lemma. Moreover, using
pseudo-resolvent identity (\ref{eq-peudoresolvent2}) we get that
$Q_{\lambda }- Q_{\lambda_L }\to 0 $ for $\lambda \to \lambda_L$
and the convergence is understood in the norm sense. Since $\alpha
$ is an eigenvalue of $Q_{\lambda_L}$ we get the final claim.
\end{proof}

\medskip

Relying on the last lemma and \ref{existenceBS} we state that
the eigenvalue of $H_{\alpha ,\Gamma_\epsilon }$ approaches the
eigenvalue of $\Hag$.
Furthermore, for $\epsilon $ and $\lambda -\lambda _L
$ small enough we can introduce the eigenprojector $P^\epsilon
_\lambda$ onto the spaces spanned by the eigenvectors of $\Qle$
corresponding to $\eta (\lambda , \epsilon )$. In the following we
will use the representation of $P^\epsilon _\lambda $ by means of
the resolvent of $Q_{\lambda }^\epsilon$, i.e.
\begin{equation}\label{projector}
P^\epsilon _\lambda =\frac{1}{2\pi i}\oint_C
R^{\epsilon}_{\lambda}(z)\,\mathrm{d}z \quad \mathrm{with }\quad
R^{\epsilon}_{\lambda }(z):=(Q_{\lambda }^\epsilon -z)^{-1}
\end{equation}
and $C:=\{\alpha +r\e ^{i\varphi}:\: \varphi\in [0,2\pi)\,,
0<r<|\alpha |\}$. Furthermore, $R^{\epsilon}_{\lambda }(z)$
satisfies a first-resolvent-type identity of the following form
\begin{equation}\label{1resolvent}
R^{\epsilon}_{\lambda }(z)=\mathcal{I}_\epsilon R_{\lambda }(z)
\mathcal{I}^c _\epsilon +R^{\epsilon}_{\lambda
}(z)(\mathcal{I}_\epsilon Q_\lambda -Q_\lambda ^\epsilon
\mathcal{I}_\epsilon )R_{\lambda }(z) \mathcal{I}^c _\epsilon\,.
\end{equation}
According to the previous discussion the eigenvalue $\lambda
(\epsilon )$ is a zero of the function $\eta (\lambda ,\epsilon
)-\alpha $ by~(\ref{existenceBS}), i.e. we have $\eta (\lambda
(\epsilon ), \epsilon )-\alpha =0$. Thus to derive the asymptotics
of $\lambda (\epsilon )$ the most natural way is employ the
implicit function theorem which requires to know the asymptotics
of $\eta (\lambda ,\epsilon )$. Let us note that
\begin{equation}\label{eq-decompeta}
\eta (\lambda ,\epsilon )= (\Qle \Ple \mathcal{I}_\epsilon \phi ,
 \Ple \mathcal{I}_\epsilon \phi ) \| \Ple \mathcal{I}_\epsilon
 \phi\|^{-2}\,,
\end{equation}
where $\phi \in \ker (Q_{\lambda _L}-\alpha )$. To recover the
asymptotics of $\eta (\lambda ,\epsilon )$ we write it as
\begin{eqnarray}\label{}
\eta (\lambda ,\epsilon )=A(\lambda ,\epsilon)+B (\lambda
,\epsilon)+C(\lambda ,\epsilon) -\alpha \,,
\end{eqnarray}
where $A(\lambda ,\epsilon):= \eta (\lambda ,\epsilon )-(\Qlec
\phi , \phi)$, $B(\lambda ,\epsilon):=(\Qlec \phi , \phi
)-(Q_{\lambda }\phi , \phi )$, and $C(\lambda
,\epsilon):=(Q_{\lambda }\phi , \phi )- (Q_{\lambda_L }\phi , \phi
)$. The asymptotics of $B(\lambda ,\epsilon)$ was already derived
in Lemma~\ref{th-estimate1}, now we want to find the asymptotics
of $A(\lambda ,\epsilon)$. To this aim we first prove the
following lemma.
\begin{lemma} \label{th-estimate1a}
As $\epsilon\to 0$ and $\lambda -\lambda_L \to 0$, we have the
relation
\begin{equation}\label{aux1est}
\|(P^\epsilon_\lambda  -I)\mathcal{I}_\epsilon \phi \|
=\mathcal{O}(\epsilon \ln \epsilon )+\mathcal{O}(\lambda -\lambda_L )\,.
\end{equation}
\end{lemma}
\begin{proof}
Applying (\ref{projector}), (\ref{1resolvent}) and using the fact
that $\mathcal{I}^c _\epsilon \mathcal{I}_\epsilon \phi=\chi^c
_\epsilon \phi$ we get by a straightforward calculation
\begin{eqnarray}
&& \lefteqn{ \|(P^\epsilon _\lambda -I)\mathcal{I}_\epsilon \phi
\| } \label{proj2} \\  &&\leq \|\mathcal{I}_\epsilon ( P_\lambda
\chi^c _\epsilon -I) \phi \|+ \frac{1}{2\pi} \oint_C
\|R^{\epsilon}_{\lambda }(z)(\mathcal{I}_\epsilon Q_\lambda
-Q_\lambda ^\epsilon \mathcal{I}_\epsilon )R_{\lambda }(z) \chi
^c_\epsilon \phi \||\mathrm{d}z|\,. \nonumber
\end{eqnarray}
To handle the first r.h.s. term in (\ref{proj2}) we employ the
triangle inequality,
\begin{equation}\label{1a-decomposition}
\|\mathcal{I}_\epsilon ( P_\lambda  \chi^c _\epsilon - I)\phi
\|\leq \|\mathcal{I}_\epsilon (P_\lambda  - P_{\lambda_L}) \chi^c
_\epsilon \phi \|+\|\mathcal{I}_\epsilon (P_{\lambda_L}  \chi^c
_\epsilon -I) \phi\|\,.
\end{equation}
Using the pseudo-resolvent formula (\ref{eq-peudoresolvent}) and
the representation of the projector by means of the resolvent we
get $\|\mathcal{I}_\epsilon (P_\lambda  - P_{\lambda_L}) \chi^c
_\epsilon \phi \|=\mathcal{O}(\lambda -\lambda_L )$. Moreover,
since $P_{\lambda _L}$ is the eigenprojector onto the space
spanned by $\phi$ we have $\|\mathcal{I}_\epsilon (P_{\lambda_L}
\chi^c _\epsilon -I) \phi\|=\mathcal{O}(\epsilon )$.
To estimate the second term of (\ref{proj2}) we consider
$\|(\mathcal{I}_\epsilon Q_\lambda -Q_\lambda ^\epsilon
\mathcal{I}_\epsilon )f\|$ where $f\in D(Q_\lambda )\cap
W^{2,1}(I)$. Using (\ref{decompositionQ}), (\ref{definitionJ}) and
the results of Lemma~\ref{th-estimate1} we obtain
\begin{equation}\label{errorest}
\|(\mathcal{I}_\epsilon Q_\lambda -Q_\lambda ^\epsilon
\mathcal{I}_\epsilon )f\|=\|Jf\|=|f(s_0 )|\mathcal{O}( \epsilon
\ln \epsilon )\,.
\end{equation}
Moreover, let us note that the function $g=R_{\lambda }(z)\chi^c
_\epsilon \phi $ belongs to $ W^{2,1}(I) $. Indeed, to see this consider
$(Q_\lambda -z )g$ which is a function from $W^{1,2}(I)$ because
$\chi^c _\epsilon \phi \in W^{1,2}(I)$
by~Lemma~\ref{le-reguarity}. Now we can repeat the arguments from
Lemmata~\ref{le-regularityD} and \ref{le-reguarity}, i.e. we have
$\mathcal{D}_\lambda  ^c  g \in W^{1,2}(I)$, and therefore $(T^c
_\lambda  -z ) g\in W^{1,2}(I)$, so finally
$$
\|D g \|\leq C\|D(T^c _\lambda  -z )
g \| < \infty\,; 
$$
see (\ref{e-estDphi}). Since $g\in W^{2,1}(I) $ it makes sense
to consider $g(s_0)$ and to employ (\ref{errorest}). Consequently,
the second term in (\ref{proj2}) can be estimated as
\begin{equation}\label{2termerror}
\|R^{\epsilon}_{\lambda }(z)(\mathcal{I}_\epsilon Q_\lambda
-Q_\lambda ^\epsilon \mathcal{I}_\epsilon )g \|\leq
\frac{1}{r}\|(\mathcal{I}_\epsilon Q_\lambda -Q_\lambda ^\epsilon
\mathcal{I}_\epsilon )g \|=\mathcal{O}(\epsilon \ln \epsilon )\,,
\end{equation}
where $r=|z-\alpha |$. Combining these estimates we get the sought
claim.
\end{proof}

\medskip

The asymptotics for $A(\lambda ,\epsilon )$ is given in the
following lemma.
\begin{lemma} \label{th-estimate2}
In the limits $\epsilon\to 0$ and $\lambda -\lambda_L \to 0$ we
have
$$
|A(\lambda ,\epsilon )| =|\eta (\lambda ,\epsilon )-(Q_{\lambda
}^\epsilon \mathcal{I}_\epsilon \phi , \mathcal{I}_\epsilon \phi
)|=o(\epsilon \ln \epsilon)+ \mathcal{O}((\lambda -\lambda_L)^2)
+\mathcal{O}(\epsilon \ln \epsilon)\mathcal{O}(\lambda
-\lambda_L)\,.
$$
\end{lemma}
\begin{proof}
Let us note that using the properties of the eigenprojector and the
asymptotics $\|P_{\lambda }^{\epsilon}\mathcal{I}_\epsilon \phi  \|=
1+\mathcal{O}(\epsilon \ln \epsilon)+ \mathcal{O}(\lambda -\lambda_L)$
which is a consequence of the previous lemma we can estimate
\begin{eqnarray}
\lefteqn{  |A(\lambda ,\epsilon )|=  |(Q_{\lambda }^\epsilon
P^\epsilon _{\lambda }\mathcal{I}_\epsilon \phi , P^\epsilon
_{\lambda }\mathcal{I}_\epsilon \phi )\|P^\epsilon _{\lambda
}\mathcal{I}_\epsilon \phi \|^{-2} -(Q_{\lambda }^\epsilon
\mathcal{I}_\epsilon \phi , \mathcal{I}_\epsilon \phi )|} \nonumber \\
\label{estprojector} && \leq \|Q_{\lambda }^\epsilon (P^\epsilon
_{\lambda }-I)\mathcal{I}_\epsilon \phi \|\|(P^\epsilon _{\lambda
} -I)\mathcal{I}_\epsilon \phi \| (1+\mathcal{O}(\epsilon \ln
\epsilon )+ \mathcal{O}(\lambda -\lambda_L)) \,. \phantom{AAAAAAA}
\end{eqnarray}
The asymptotics for $\|(P^\epsilon _{\lambda }
-I)\mathcal{I}_\epsilon \phi \|$ was explicitly derived in
Lemma~\ref{th-estimate1a}. Furthermore,
proceeding in analogy with (\ref{proj2}) we
find
\begin{eqnarray}
\lefteqn{ \|\Qle (P^\epsilon _\lambda -I)\mathcal{I}_\epsilon \phi
\| \leq \|\Qle \mathcal{I}_\epsilon ( P_\lambda  \chi^c _\epsilon
-I) \phi \|} \nonumber \\  && + \label{aux-b} \frac{1}{2\pi}
\oint_C \|\Qle R^{\epsilon}_{\lambda }(z)(\mathcal{I}_\epsilon
Q_\lambda -Q_\lambda ^\epsilon \mathcal{I}_\epsilon )R_{\lambda
}(z) \chi ^c_\epsilon \phi \||\mathrm{d}z|\,.
\end{eqnarray}
Mimicking now the argument of (\ref{1a-decomposition}) we estimate
the first term on the r.h.s. of (\ref{aux-b}) obtaining
\begin{equation}\label{aux1}
\|\Qle \mathcal{I}_\epsilon ( P_\lambda  \chi^c _\epsilon - I)\phi
\|\leq \|\Qle  \mathcal{I}_\epsilon (P_\lambda  - P_{\lambda_L})
\chi^c _\epsilon \phi \|+\|\Qle \mathcal{I}_\epsilon
(P_{\lambda_L} \chi^c _\epsilon -I) \phi\|\,.
\end{equation}
Furthermore
\begin{eqnarray}
\lefteqn{\nonumber \|\Qle  \mathcal{I}_\epsilon (P_\lambda  -
P_{\lambda_L}) \chi^c _\epsilon \phi \|\leq \| ( Q_{\lambda
}^\epsilon \mathcal{I}_\epsilon- \mathcal{I}_\epsilon Q_\lambda)
(P_\lambda - P_{\lambda_L}) \chi^c _\epsilon \phi \|} \\ &&
\phantom{AAA} + \| \mathcal{I}_\epsilon Q_\lambda (P_\lambda -
P_{\lambda_L}) \chi^c _\epsilon \phi \|\,, \phantom{AAAAAAAAAAAA}
\end{eqnarray}
where $\| ( Q_{\lambda }^\epsilon \mathcal{I}_\epsilon-
\mathcal{I}_\epsilon Q_\lambda) (P_\lambda - P_{\lambda_L}) \chi^c
_\epsilon \phi \| = \mathcal{O}(\epsilon\ln \epsilon
)\mathcal{O}(\lambda -\lambda_L )$ and $\| \mathcal{I}_\epsilon
Q_\lambda (P_\lambda - P_{\lambda_L}) \chi^c _\epsilon \phi
\|=\mathcal{O}(\lambda -\lambda_L )+\mathcal{O} (\epsilon )$. Proceeding
analogously as with the second term of (\ref{aux1}) we get $\|\Qle
\mathcal{I}_\epsilon (P_{\lambda_L} \chi^c _\epsilon -I)
\phi\|=\mathcal{O}(\epsilon\ln \epsilon )$, and therefore
\begin{equation}\label{e-estQle}
\|\Qle \mathcal{I}_\epsilon ( P_\lambda  \chi^c _\epsilon - I)\phi
\|=\mathcal{O}(\epsilon \ln \epsilon)+ \mathcal{O}(\lambda -\lambda_L)
+\mathcal{O}(\epsilon \ln \epsilon)\mathcal{O}(\lambda -\lambda_L)\,.
\end{equation}
To handle the second term in (\ref{aux-b}) let us note that
$$\|Q_{\lambda }^\epsilon R^{\epsilon}_{\lambda }(z)\|\leq
1+\frac{|z|}{r}\leq 2+\frac{|\alpha|}{r}\,,$$
hence using (\ref{2termerror}) we obtain
$$\|Q_{\lambda }^\epsilon R^{\epsilon}_{\lambda
}(z)(\mathcal{I}_\epsilon Q_\lambda -Q_\lambda ^\epsilon
\mathcal{I}_\epsilon )R_{\lambda }(z) \chi^c _\epsilon \phi
\|=\mathcal{O}(\epsilon \ln \epsilon )\,,
$$
which finally gives
\begin{equation}\label{auxest2}
\|Q_{\lambda }^\epsilon (P_\lambda ^\epsilon -1)\mathcal{O}_\epsilon \phi
\|=\mathcal{O}(\epsilon \ln \epsilon)+ \mathcal{O}(\lambda -\lambda_L)
+\mathcal{O}(\epsilon \ln \epsilon)\mathcal{O}(\lambda -\lambda_L)\,.
\end{equation}
Putting the above results together and applying
Lemma~\ref{th-estimate1a} to (\ref{estprojector}) we get the claim
of the lemma.
\end{proof}

\medskip

Putting the results of Lemmata~\ref{th-estimate1},
\ref{th-estimate1a}, \ref{th-estimate2} together and applying
(\ref{eq-decompeta}) we get
\begin{eqnarray} \nonumber
&& \eta (\lambda ,\epsilon )= \frac{2}{\pi}|\phi (s_0 )|^2
\epsilon \ln \epsilon +(Q_\lambda \phi ,\phi)+
\\ \label{eq-formeta} && o(\epsilon \ln
\epsilon)+ \mathcal{O}((\lambda -\lambda_L)^2)
+\mathcal{O}(\epsilon \ln \epsilon)\mathcal{O}(\lambda
-\lambda_L)\,,
\end{eqnarray}
as the hiatus half-length $\epsilon$ and the eigenvalue difference
$\lambda -\lambda_L$ tend to zero.

Let us keep the notation $\lambda_L$ for the eigenvalue of $\Hag$
which means that $\mathrm{ker}\,(Q_{\lambda_L}-\alpha)$ is
nontrivial and suppose as before that $\phi \in
\mathrm{ker}\,(Q_{\lambda_L}-\alpha)$ is the normalized function
in $L^{2}(I)$. Our goal is to find an asymptotic expression for
the eigenvalue of $H _{\Gamma _{\epsilon},\alpha }$ by means of
$\lambda_L$ and $\phi$.
\begin{theorem}\label{ev.expansion}
The eigenvalue of $\Hag$ admits the following asymptotic expansion
as $\epsilon\to 0$,
\begin{equation}\label{expansion-kappa}
\lambda (\epsilon )=\lambda_{L}-\omega (\kappa_L )|\phi (s_0)|^2
\epsilon \ln \epsilon +o(\epsilon \ln \epsilon )\,,  \end{equation}
where
$$
\omega (\lambda _L )=16 \kappa_{L}\left(\int_{I\times
I}\e^{-\kappa_L |\gamma(s)-\gamma(t)|}
\phi(s)\overline{\phi(t)}\,\mathrm{d}s
\mathrm{d}t\right)^{-1}\,,\quad \kappa_L:= \sqrt{-\lambda_L}\,.
$$
\end{theorem}
\begin{proof}
Due to (\ref{existenceBS}) the eigenvalue $\lambda(\epsilon )$ is
determined by the condition $\mathrm{ker}\,(Q_{\lambda (\epsilon
)}^{\epsilon}-\alpha)\neq \{0\}$. It is convenient to put
$$
\hat{\eta }( \lambda ,\delta )\equiv \eta (\lambda
,\epsilon)-\alpha  :U_0 \times \C \mapsto \C\, \quad \mathrm{where }\quad
\delta :=\epsilon \ln \epsilon
$$
and $U_0$ is a neighborhood of zero. Our aim is to find where
the function $\hat{\eta}$ vanishes. Using the fact that $\hat{\eta} (
\lambda_{L} ,0)=0$ and $\hat{\eta }\in C^1 \times C^{\infty }$ and
relying on the implicit function theorem we can evaluate
$$
\lambda(\epsilon )=\lambda_{L}-(\partial_\delta \hat{\eta} )\mid
_{\theta_L}(
\partial _ \lambda \hat{\eta} )^{-1}\mid_{\theta_L}
\delta +o(\delta)\,,\,\quad \theta_L \equiv(\lambda_{L},0)\,.
$$ 
To find $\partial_{\delta}\hat{\eta} \mid _{\theta_L}$ we use the
asymptotics (\ref{eq-formeta})
\begin{eqnarray}\label{deriv-epsilon}
\nonumber &&\frac{1}{\delta}\Big(\hat{\eta} (\lambda_L ,\delta
)-\hat{\eta} (\lambda_L ,0)\Big) \to \frac{2}{\pi }|\phi (s_0)|^2
\quad \mathrm{as} \quad \delta\to 0 \,.\label{derivation1}
\end{eqnarray}
To find the other derivative we use (\ref{eq-formeta}) to state
$$
( \partial _ \lambda \hat{\eta} )\mid_{\theta_L}=(\partial _
\lambda (Q_\lambda \phi ,\phi))\mid_{\theta_L}\,.
$$
On the other hand the derivative of $Q_{\lambda}$ w.r.t. $\lambda$
coincides with the derivative of $G_{\lambda}$ because the
regularization we made was independent of the spectral parameter
$\lambda$; therefore we have
\begin{equation}\label{derivation2}
(\partial_{\lambda}Q_{\lambda}(s,t) )
\mid_{\theta_{L}}=\frac{1}{8\pi \kappa_{L}} e^{-\kappa_L
|\gamma(s)-\gamma(s)|}\,.
\end{equation}
Putting together (\ref{derivation1}), (\ref{derivation2}) we get
the sought result.
\end{proof}

\bigskip

As the final step of is this section we return to the general
question and extend the above theorem to the case when $\Hag $
have more than one eigenvalue; recall that since $\Gamma$ is
finite by assumption we have $\sharp\sigma_\mathrm{d}(\Hag)<
\infty$. Suppose that $\lambda_L ^1 <\lambda_L ^2 \leq ...\leq
\lambda_L ^N $, $N\in \N$ are the eigenvalues of $\Hag$ and
$\{\phi _i\}_{i=1}^N$ is the corresponding eigenfunction system
which is assumed to be normalized. Given $\lambda _L\in
\sigma_{\mathrm{d}}(\Hag )$ define
\begin{eqnarray}\nonumber
m(\lambda_L ):=\min \{j=1,...,N \,:\, \lambda_j =\lambda_L \}\,,\quad
\nonumber n(\lambda_L ):=\max \{j=1,...,N \,:\, \lambda_j
=\lambda_L \}
\end{eqnarray}
and the matrix $C(\lambda_ L )$ given by
$$
[C(\lambda_ L )]_{ij}:=\phi_i (s_0 )\overline{\phi_j (s_0
)}\omega_{ij}\,,\quad i,j =m(\lambda_L ),...,n(\lambda_L )\,,
$$
where
$$
\omega_{ij}(\lambda _L ):=
\left(\int_{I\times I}\e^{-\kappa_L |\gamma(s)-\gamma(t)|} \phi_i
(s)\overline{\phi_j (t)}\,\mathrm{d}s \mathrm{d}t\right)^{-1}\,.
$$
Using this notation we can state our main result:
\begin{theorem}\label{th-final}
Let $\lambda_{L} \in\sigma_{\mathrm{d}}(\Hag )$. Then the
corresponding eigenvalues of $H_{\alpha, \Gamma_{\epsilon}}$ have
the following asymptotic expansion,
$$
\lambda _j (\epsilon )=\lambda_{L} -s_j (\lambda_L )\epsilon \ln
\epsilon +o(\epsilon \ln \epsilon )\,,\quad m(\lambda_L)\le j \le
n(\lambda_L)\,,
$$
as $\epsilon\to 0$, where $s_j(\lambda _L )$ are the eigenvalues
of matrix $C (\lambda_L )$.
\end{theorem}
The proof of essentially repeats the reasoning used above; the
only new element is that different eigenfunctions corresponding to
the same eigenvalue $\lambda_L$ correspond to the appropriate
scalar products. This consequently leads to the appearance of the
matrix $C(\lambda _L)$ which reduces to $|\phi (s_0 )|^2\omega
(\lambda_L )$ if $\lambda_L$ is a simple eigenvalue.

\section{Concluding remarks}
First we note that the hiatus perturbation of a curve in $\R^3$
can be regarded as an effective repulsive interaction. The
presence of a hiatus pushes the eigenvalues up which can be easily
seen from (\ref{expansion-kappa}) since $\omega (\lambda_L)>0$ and
$\epsilon \ln \epsilon <0 $ for small enough $\epsilon$. This
might be expected, of course, because the interaction supported by
a curve in $\R^3 $ is attractive as manifested by the fact that it
produces bound states, at least if the curve is sufficiently long,
cf.~Theorem~\ref{th-exbound}.

\smallskip

Comparing the eigenvalue asymptotics (\ref{expansion-kappa}) with
the analogous result for a curve in $\R^2$ derived in \cite{EY2}
we can see the difference in the first asymptotic term, which in
the codimension one case behaves as $\mathcal{O}(\epsilon)$ in
contrast to $\mathcal{O}(\epsilon \ln \epsilon )$ obtained here.
The former result is a natural consequence of the additive
character of the singular potential manifested by the sum-type
quadratic form representation of the corresponding Hamiltonian.
Such a representation does not exists if the potential is
supported by a set of codimension two. To find a self-adjoint
realization of the $\delta$ interaction in this case we have to
perform, for instance, a logarithmic regularization of the
appropriate quantities, and consequently, the eigenvalue
asymptotics w.r.t. the length of the hiatus, as well as its
derivation, are more involved.
\setcounter{equation}{0}
\section{Appendix: the remaining proofs} \label{appendix}
To prove Theorem~\ref{th-ear} we need the following lemma.
\begin{lemma} \label{le-ear}
Given $s\in [0,L]$ corresponding to
$\tilde{\gamma}_{d,s}=\tilde{\gamma}_{d}$ and $d>0$, and making
$|s-t|$ small we have
$$
|\gamma (s)-\tilde{\gamma }_d (t)|^2=(s-t)^2 (1+\Oo (d))+d^2 +\Oo
((s-t)^3)\,.
$$
\end{lemma}
\begin{proof}
An elementary cosine formula gives
$$
|\gamma (s)-\tilde{\gamma }_d (t)|^2= |\tilde{\gamma} _d
(s)-\tilde{\gamma }_d (t)|^2+d^2 -2 \iota (s,t)\,,
$$
where $\iota (s,t):=(\gamma (s)-\tilde{\gamma }_d (s)\, ,\,
\tilde{\gamma} _d (s)-\tilde{\gamma }_d (t))$. Note that $\iota
(s,t)|_{t=s}=0 $ and $\partial _t \iota (s,t)|_{t=s}=0 $ holds in
view of the assumption $\mathrm{(\widetilde b3)}$. Furthermore,
the Taylor expansion in the corresponding ``shifted'' points of
the coordinate projections of the curve $\tilde{\gamma }_d$ yields
$$
|\tilde{\gamma }_d(s)-\tilde{\gamma }_d (t)|^2= \sum_{i=1} ^{3}
\dot{\tilde{\gamma }}_{d,i} (\theta _i )^2 (s-t)^2 +\Oo ((s-t)^3
\,.
$$
Using the Taylor expansion again and combining it with the
asymptotics given by $\mathrm{(\widetilde b1)}$,
$\mathrm{(\widetilde b2)}$ and the fact that $\sum_{i=1}
^{3}\dot{\gamma }_{i} (s)^2=1$ we get the claim.
\end{proof}

\bigskip

\noindent {\bf Proof of Theorem~\ref{th-ear}}. The first step is
show that in the following limits
\begin{equation}\label{eq-1limit}
\lim_{d\to 0 }\int_{I}\left[ G_z ( \gamma (s)-\tilde{\gamma }_d
(t) ) -\frac{1}{4\pi |\gamma (s)-\tilde{\gamma }_d (t)|}\right]
f(t)\,\D t
\end{equation}
and
\begin{equation}\label{eq-2limit}
\lim_{d\to 0 }\int_{I}\left[ \frac{1}{ |\gamma (s)-\tilde{\gamma
}_d (t)|}- \frac{1}{ ((s-t)^2 +d^2)^{1/2}}\right] f(t)\,\D t
\end{equation}
we can interchange the limit with the integration. Using the
inequality $|(\e^{-\kappa x }-1)x^{-1}|\leq \kappa $ for $\kappa $
and $x $ positive we can use the dominated convergence to prove
claim concerning (\ref{eq-1limit}). To handle the second limit we
can use Lemma~\ref{le-ear} and show that
\begin{eqnarray}
\lefteqn{\nonumber \frac{1}{|\gamma (s)-\tilde{\gamma }_d (t)|}-
\frac{1}{((s-t)^2 +d^2)^{1/2}}}\\ [.5em]  \label{e-estgamma}&&
=\Big((s-t)^2 \Oo (d)+ \Oo((s-t)^3) \Big)\Big((s-t)^2 +d^2 \Big)^{-1}\leq
\mathrm{const}\,.
\end{eqnarray}
Therefore using the dominated convergence again we can perform the
interchange in (\ref{eq-2limit}).
The resulting limit of the sum of both the expressions
(\ref{eq-1limit}) and (\ref{eq-2limit}) is given by
\begin{eqnarray}
\nonumber\lim_{d\to 0 }\int_{I}\left[ G_z ( \gamma
(s)-\tilde{\gamma }_d (t) )- \frac{1}{4\pi ((s-t)^2
+d^2)^{1/2}}\right] f(t)\,\D t
\\ \label{eq-earlimit1} = \int_{I}\mathcal{A}_{z}(s-t)f(t)\mathrm{d}t
+\int_{I}\mathcal{D}_{z}(s,t)f(t) \,\mathrm{d}t\,,
\end{eqnarray}
where $\mathcal{A}_{z}$ and $ \mathcal{D}_{z}$ are defined in
Lemma~\ref{formQ}. Repeating the argument from the proof of this
lemma, see (\ref{Idelta}) and (\ref{-Idelta}), we get
\begin{eqnarray}
\lefteqn{\lim_{d\to 0}\left[\int_{I_\delta} \frac{1}{4\pi ((s-t)^2
+d^2 )^{1/2}} f(t)\,\mathrm{d}t +\frac{1}{2\pi}f(s)\ln d \right]}
\nonumber \\ &&  =
\frac{1}{4\pi}\bigg(\int_{I}\frac{f(t)-f(s)}{|t-s|}\,\mathrm{d}t
+\ln 4 s(L-s)f(s)\bigg)\,. \label{eq-earlimit2}
\end{eqnarray}
The final step is to note that
$$
\int_{I} G_z ( \gamma (s)-\tilde{\gamma }_d (t) )(1-j_d (t))f(t)
\,\D t =o(1)
$$
as $d\to 0$, because $j_d= 1+\Oo (d)$ and $\int_{I} G_z ( \gamma
(s)-\tilde{\gamma }_d (t) )f(t)\,\D t$ has a singularity of the
type $f(s)\ln d $. Combing this with (\ref{eq-earlimit1}) and
(\ref{eq-earlimit2}) we conclude the proof of
Theorem~\ref{th-ear}

\medskip

\begin{remark} \label{re-LDCTh}
\rm{Using the same arguments as in (\ref{e-estgamma}) we can
estimate
\begin{eqnarray}
\lefteqn{\nonumber \frac{1}{|\gamma _d (s)-\gamma  (t)|}-
\frac{1}{((s-t)^2 +d^2)^{1/2}}} \\ &&  =\Big((s-t)^2 \Oo (d)+
\Oo((s-t)^3 )\Big) \Big((s-t)^2 +d^2 \Big)^{-1}\leq \mathrm{const}\,,
\end{eqnarray}
which directly implies
$$
|\mathcal{D }_z^{d} (s,t)|=|G_z (\gamma _d (s)-\gamma  (t))-G_z ^d (s-t)
|\leq \mathrm{const}\,,
$$
for any $s\in [0,L]$.}
\end{remark}

\bigskip

\noindent {\bf Proof of Lemma~\ref{le-regularityD}.} Recall that
our goal is to show that
\begin{equation}\label{e-claim}
\int_{I } \mathcal{D}^c _z (s,t)f(t)\,\mathrm{d}t \in
W^{1,2}\,\,\,\mathrm{for}\,\,\,f\in L^2 (I)\,,
\end{equation}
where
$$\mathcal{D}^c _z (s,t):= G_z (\gamma (s)-\gamma (t))-G_z
(\gamma ^c (s)-\gamma ^c (t))\,
$$
and $G_z (\rho )=\e^{-{\sqrt{-z}}|\rho|} (4\pi|\rho|)^{-1} $.
We proceed in three steps: \\
\emph{Step 1:} We show that the following inequality holds
\begin{equation}\label{e-ineq1}
||\gamma (s)-\gamma (t)|-|\gamma^c (s)-\gamma^c (t)||\leq c_1
|s-t|^{\mu +1}
\end{equation}
for $c_1 |s-t|^\mu <1$. By a straightforward calculation one can
find that
$$
|\gamma^c (s)-\gamma^c (t)|^2 =\frac{L^2 }{2 \pi^2}\left( 1-\cos
\frac{2\pi (s-t)}{L} \right) \,.
$$
Consequently, there exists a positive constant $\tilde{c}$ such
that
\begin{equation}\label{e-ineqgc}
|\gamma^c (s)-\gamma ^c (t)|\geq |s-t|(1-\tilde{c}|s-t|^2)
\end{equation}
for $\tilde{c} |s-t|^2 <1$. Using the above inequality we have
\begin{equation}\label{e-1}
|\gamma (s)-\gamma (t)|\leq |s-t| \leq |\gamma^c (s)-\gamma ^c
(t)|+ \tilde{c} |s-t|^3\,.
\end{equation}
On the  other hand, using the assumption $\mathrm{(a)}$ we obtain
\begin{equation}\label{e-2}
|\gamma (s)-\gamma (t)|\geq |s-t| -c|s-t|^{\mu +1}\geq |\gamma^c
(s)-\gamma ^c (t)|-c|s-t|^{\mu +1} \,.
\end{equation}
Combining (\ref{e-1}) and (\ref{e-2}) we arrive at
(\ref{e-ineq1}).
\noindent \emph{Step 2:} The aim of this part of the proof is to
show the following asymptotics,
\begin{equation}\label{e-asymDc}
\mathcal{D}^c _z (s,t) = \mathcal{O}(|s-t|^{\mu -1})\,.
\end{equation}
Using (\ref{e-ineq1}), (\ref{e-ineqgc}) and the assumption
$\mathrm{(a)}$ we get
\begin{eqnarray}
\lefteqn{ \nonumber |T(s,t)|:= \left|
\frac{1}{|\gamma(s)-\gamma(t)|}-\frac{1}{|\gamma^c (s)-\gamma ^c
(t)|}\right|} \nonumber \\&& \leq \frac{c_1  |s-t|^{\mu
+1}}{|s-t|^2 (1-c|s-t|^\mu )(1-\tilde{c}|s-t|^2 )}\leq c_2
|s-t|^{\mu -1}\,, \label{e-estT}
\end{eqnarray}
where $c_2 $ is a positive constant. Furthermore, using the the
exponential function expansion and (\ref{e-ineq1}) we find
$$
\mathcal{D}^c _z (s,t)= T(s,t)+\mathcal{O}(|s-t|^{\mu +1})\,,
$$
which in view of (\ref{e-estT}) implies (\ref{e-asymDc}).
\noindent \emph{Step 3:} Let us note that for $f\in L^2 (I)$ we
have
$$
\bigg|\int_{I}\mathcal{D}^c _z (s,t)f(t) \,\mathrm{d}t \bigg|\leq
a(s) \|f\|_{L^2 (I)}\,,\quad \mathrm{where }\quad
a(s):=\int_{I}|\mathcal{D}^c _z (s,t)|^2 \mathrm{d}t\,.
$$
Using (\ref{e-asymDc}) we claim that $a'(s)$ is an integrable
function, and therefore we can use the dominated convergence to
show that
\begin{eqnarray}
\lefteqn{\nonumber \int_{I} \Big|D_s \int_{I}\mathcal{D}^c _z
(s,t)f(t) \,\mathrm{d}t \Big|^2 \mathrm{d}s=  \int_{I} \Big|
\int_{I}D_s \mathcal{D}^c _z (s,t)f(t) \,\mathrm{d}t \Big|^2
\,\mathrm{d}s}
\\ && \leq \int_{I} \int_{I}|D_s \mathcal{D}^c _z (s,t) |^2
\,\mathrm{d}t \,\mathrm{d}s \,\|f\|_{L^2 (I)}< \infty \,,
\phantom{AAAAAAAA} \label{e-auxest}
\end{eqnarray}
where $D_s$ stands for the derivative; in the above estimates we
have again used (\ref{e-asymDc}) to check that the last term in
the cahin (\ref{e-auxest}) is finite. This finally proves
(\ref{e-claim}), and by that Lemma~\ref{le-regularityD}.

\medskip

\begin{remark} \label{re-estR}
{\rm  Let us note that in analogy with (\ref{e-estT}) we can
estimate
\begin{eqnarray}\label
&&  \left| \frac{1}{|\gamma(s)-\gamma(t)|}-\frac{1}{|s-t|}\right|
\leq \frac{c |s-t|^{\mu +1}}{|s-t|^2 (1-c|s-t|^\mu )}\leq c_3
|s-t|^{\mu -1}\,,
\end{eqnarray}
where we have again used the assumption $\mathrm{(a)}$. This
implies
$$
|\mathcal{R}_{\lambda }(s,t)|=|G_\lambda (\gamma (s)-\gamma
(t))-(4\pi |s-t|)^{-1}| \leq c_4 |s-t|^{\mu -1}\,.
$$
}
\end{remark}

\subsection*{Acknowledgments}
The research was partially supported by Ministry of Education,
Youth and Sports of the Czech Republic under the project LC06002.
S.K. thanks to her son Antek that he had allowed her to work on
this paper in the first months of his life.


\begin{thebibliography}{99}
\setlength{\itemsep}{-4pt}
\bibitem[AGHH]{AGHH}
S.~Albeverio, F.~Gesztesy, R.~H\o egh-Krohn, H.~Holden: {\em
Solvable Models in Quantum Mechanics}, 2nd printing, AMS,
Providence, R.I., 2004.
\bibitem[AGS]{AGS}
J.-P.~Antoine, F.~Gesztesy, J.~Shabani: Exactly solvable models of
sphere interactions in quantum mechanics, {\em J. Phys.} {\bf A20}
(1987), 3687-3712.
\bibitem[BCFK]{BCFK}
G.~Berkolaiko, R.~Carlson, S.~Fulling, P.~Kuchment, eds.:
\emph{Quantum Graphs and Their Applications}, Contemporary Math.,
vol.~415, AMS, Providence, R.I., 2006.
\bibitem[BEK\v{S}]{BEKS}
J.F.~Brasche, P.~Exner, Yu.A.~Kuperin, P.~\v{S}eba: Schr\"odinger
operators with singular interactions, {\em J. Math. Anal. Appl.}
{\bf 184} (1994), 112-139.
\bibitem[BO]{BO}
J.F.~Brasche, K.~O\v{z}anov\'{a}: Convergence of Schr\"odinger
operators, \emph{SIAM J. Math. Anal.} {\bf 39} (2007), 281-297.
\bibitem[BT]{BT}
J.F.~Brasche, A.~Teta: Spectral analysis and scattering theory for
Schr\"odinger operators with an interaction supported by a regular
curve, in {\em Ideas and Methods in Quantum and Statistical
Physics}, ed. by S.~Albeverio, J.E.~Fenstadt, H.~Holden,
T.~Lindstr\o m, Cambridge Univ. Press 1992, pp.~197-211.
\bibitem[Ex]{Ex}
P.~Exner: Leaky quantum graphs: a review, a contribution to
\cite{EKST}; \texttt{arXiv: 0710.5903 [math-ph]}
\bibitem[EF]{EF}
P.~Exner, R.~Frank: Absolute continuity of the spectrum for
periodically modulated leaky wires in $\mathbb{R}^3$, {\em Ann. H.
Poincar\'{e}} {\bf 8} (2007), 241-263.
\bibitem[EI]{EI}
P.~Exner, T.~Ichinose:
Geometrically induced spectrum in curved leaky wires,
{\em J. Phys.} {\bf A34} (2001), 1439-1450.
\bibitem[EKST]{EKST}
P.~Exner, J.~Keating, P.~Kuchment, T.~Sunada, A.~Teplyaev, eds.:
\emph{Analysis on Graphs and Applications}, Proceedings of an
Isaac Newton Institute programme, AMS, a volume in preparation
\bibitem[EK02]{EK1}
P.~Exner, S.~Kondej: Curvature-induced bound states for a $\delta$
interaction supported by a curve in $\mathbb{R}^3$, {\em Ann.
H.~Poincar\'{e}} {\bf 3} (2002), 967-981.
\bibitem[EK05]{EK2}
P.~Exner, S.~Kondej:
Scattering by local deformations of a straight leaky wire,
{\em J. Phys.} {\bf A38} (2005), 4865-4874.
\bibitem[EN]{EN}
P.~Exner, K. N\v{e}mcov\'{a}: Leaky quantum graphs: approximations
by point interaction Hamiltonians, {\em J. Phys.} {\bf A36}
(2003), 10173-10193.
 \bibitem[EY02]{EY1}
P.Exner, K.Yoshitomi: Asymptotics of eigenvalues of the
Schr\"odinger operator with a strong $\delta$-interaction on a
loop, \emph{J. Geom. Phys.} {\bf 41} (2002), 344--358.
\bibitem[EY03]{EY2}
P.~Exner, K.~Yoshitomi: Eigenvalue asymptotics for the
Schr\"odinger operator with a $\delta$-interaction on a punctured
surface, {\em Lett. Math. Phys.} {\bf 65} (2003), 19-26; erratum
{\bf 67} (2004), 81-82.
\bibitem[Ka]{Ka}
T.~Kato: {\em Perturbation Theory for Linear Operators}, 2nd
edition, Springer, Berlin 1976.
\bibitem[Po04]{Po04}
A.~Posilicano: Boundary triples and Weyl Functions for singular
perturbations of self-adjoint operator, {\em Methods of Functional
Analysis and Topology} {\bf 10} (2004),  57-63.
\bibitem[Po01]{Po01}
A.~Posilicano: A Krein-like formula for singular perturbations of
self-adjoint operators and applications, {\em J. Funct. Anal. }
{\bf 183} (2001), 109-147.
\bibitem[RS]{RS}
M.~Reed and B.~Simon: {\em Methods of Modern Mathematical Physics,
II.~Fourier Analysis, Self-Adjointness, IV. Analysis of
Operators}, Academic Press, New York 1975, 1978.
\bibitem[Sha]{Sha}
J.~Shabani: Finitely many delta interactions with supports on
concentric spheres, {\em J. Math. Phys.} {\bf 29} (1988), 660-664.
\end{thebibliography}
\end{document}